\documentclass[11pt]{article}

\usepackage[margin=1in]{geometry}
\usepackage{amsmath,amssymb,amsthm,mathtools}
\usepackage{microtype}
\usepackage{graphicx}
\usepackage{amsthm,amsmath,amssymb}
\usepackage{cite}
\usepackage{tikz}
\usepackage{float}
\usepackage{subcaption}
\usepackage{hyperref}
\usepackage{cleveref}
\usepackage{lipsum}
\usetikzlibrary{patterns,intersections}
\usetikzlibrary{patterns.meta}
\usetikzlibrary{arrows.meta,positioning,calc}
\usetikzlibrary{decorations.pathreplacing}
\newtheorem{lemma}{Lemma}
\newtheorem{remark}{Remark}

\newtheorem{definition}{Definition}
\newtheorem{thm}{Theorem}
\crefname{thm}{Theorem}{Theorems}
\crefrangelabelformat{thm}{#3#1#4--#5#2#6}
\newcommand{\diag}{\mathrm{diag}} 
 
\newcommand{\ci}{\mathcal{C}^I}
\newcommand{\cf}{\mathcal{C}^F}

\newcommand{\rf}{r^F}
\newcommand{\ri}{r^I}
\newcommand{\kf}{k^F}
\newcommand{\ki}{k^I}
\newcommand{\nf}{n^F}
\newcommand{\nii}{n^I}

\newcommand{\F}{\mathbb F}
\newcommand{\col}{\operatorname{col}}
\newcommand{\rank}{\operatorname{rank}}
\newcommand{\Boverall}{B}%初始码的生成矩阵
\newcommand{\Coverall}{C}%终止码的生成矩阵
\newcommand{\Aone}{A}
\newcommand{\Mone}{M}
\newcommand{\Atwo}{A}
\newcommand{\Mtwo}{M}
\newcommand{\mthr}{m}
\newcommand{\delthr}{ \delta}
\newcommand{\hthr}{h}
\newcommand{\Tthr}{T}
\newcommand{\betthr}{\beta}
\newcommand{\sigthr}{\sigma}
\newcommand{\uthr}{u}
\newcommand{\wthr}{w}

\newcommand{\nzerthr}{n_0}
\newcommand{\dzerthr}{d_0}

\newcommand{\tauthr}{\tau}
\newcommand{\mathHthr}{\mathcal{H}}
\newcommand{\mathUthr}{\mathcal{U}}
\newcommand{\mathVthr}{\mathcal{V}}
\newcommand{\Wthr}{W}
\newcommand{\Qthr}{Q^F}
\newcommand{\Qthrini}{Q^I}

\newcommand{\Qpolthr}{Q_\mathcal{V}(X)}
\newcommand{\Fthr}{F}
\newcommand{\mathFthr}{\mathcal{F}}

\begin{document}
\title{Optimal Conversion Bandwidth for MDS Convertible Codes in the Split Regime with $\rf<\kf<\ri$}
\author{Lewen Wang \and Sihuang Hu \thanks{
    The authors are with the School of Cyber Science and Technology, Shandong University, Qingdao, Shandong 266237, China. Email: lewenwang3@gmail.com, husihuang@sdu.edu.cn. 
    % The corresponding author is Sihuang Hu.
  }
            }
\date{}
\maketitle
% {\renewcommand{\thefootnote}{}
  		% \footnotetext{
        % The corresponding author is Sihuang Hu.
  		% }
    % }
\begin{abstract}
Erasure codes are widely used in distributed storage systems to provide fault tolerance. 
An $[n,k]$ erasure code encodes $k$ data symbols into $n$ coded symbols and distributes them across $n$ storage nodes.
Once the code parameters are fixed, the achievable fault tolerance is also fixed.
However, the failure rates of storage nodes may vary over time, and dynamically adapting the code parameters to these variations can substantially reduce storage overhead.
Motivated by this observation,  Maturana and Rashmi introduced convertible codes, which allow an $[\nii,\ki]$ initial code with redundancy $\ri=\nii-\ki$ to be transformed into an $[\nf,\kf]$ final code with redundancy $\rf=\nf-\kf$, while preserving the required code properties.
Convertible codes whose initial and final codes are both MDS codes are called MDS convertible codes. 
They are of particular interest because MDS codes provide the maximum erasure tolerance for a given amount of
storage overhead.

Several works have established lower bounds and constructions for the bandwidth cost of MDS convertible codes in the split regime. 
However, the tight bound in the parameter range $\rf<\kf<\ri$  remains unknown.
% In this work, we derive an improved closed-form lower bound consisting of three cases for this remaining range by establishing a family of new rank inequalities.
In this work, we establish a family of new rank inequalities for stable MDS convertible codes with linear conversion
procedures and derive an improved lower bound on the  bandwidth cost consisting of three cases for this remaining range.
We prove that the bound is tight by presenting three explicit constructions, one for each case.
% Since the write bandwidth is fixed for stable conversion, our results also determine the minimum total conversion
% bandwidth.

\end{abstract}
\section{Introduction}

Erasure coding has become a fundamental technique for protecting  data against device failures in distributed storage systems.
In this setting, an $[n,k]$ erasure code over a finite field $\mathbb F_q$ encodes $k$ data symbols into $n$ coded symbols, which are then stored on $n$ distinct storage nodes.
Once the code parameters $n$ and $k$ are fixed, the fault-tolerance capability of the system is also fixed.
In practice, however, the failure rates of storage devices may vary over time.
In particular, the failure rate may follow a bathtub curve over the device lifetime. 
It is relatively high during an early-failure period, becomes lower and relatively stable during the useful-life period, and increases again during the wear-out period.
Prior work by Kadekodi \emph{et al.} has shown that adapting the code parameters to changes in device reliability can substantially reduce storage overhead while maintaining a prescribed reliability level~\cite{kadekodi2019cluster}.
Therefore, it is desirable to dynamically adjust the code parameters as the fault-tolerance requirements of the system change.

A straightforward approach is to decode and re-encode  all stored data, which incurs significant I/O and communication costs.
To reduce these costs,  Maturana and Rashmi introduced the framework of \emph{convertible codes}, which enables efficient conversion  from an initial code to a final code with different parameters while preserving
desired code properties~\cite{maturana2020convertible,maturana2022convertible,Maturana23}.
When additional structural properties are required, convertible codes can be designed to preserve those properties during conversion.
In particular, MDS codes are of special interest due to their widespread use in distributed storage and their optimal erasure tolerance.
A convertible code whose initial and final codes are both MDS codes is called  an \emph{MDS convertible code}.
The framework has subsequently been extended to other classes of codes, including locally repairable codes~\cite{Maturana2023isit,Kong2024,Ge2025LocallyRC} and minimum storage regenerating(MSR) codes~\cite{Yang2026MSR}.

There are two basic regimes in code conversion.
The first is the \emph{merge regime}, where several initial codewords are converted into a single final codeword.
The second is the \emph{split regime}, where one initial codeword is converted into several final codewords.
These two regimes serve as fundamental building blocks for more general code conversions.

The efficiency of code conversion is commonly measured by two  metrics.
The first is the \emph{access cost}, namely, the total amount of store data accessed during conversion.
Maturana and Rashmi established fundamental lower bounds on the access cost of MDS convertible codes and provided optimal constructions for both the merge and split regimes~\cite{maturana2020convertible,maturana2022convertible,Maturana2020isit}.
Subsequent work studied, among other directions, low-field-size constructions~\cite{Chopra2024OnLF}, irregular and generalized conversions~\cite{Ge2024MDSGC,RKSVK25,Wang2026GeneralizedMerge},
and extensions to locally repairable and regenerating codes
~\cite{Maturana2023isit,Kong2024,Ge2025LocallyRC,Shi2025BoundsAO,Yang2026MSR}.

The second metric is the \emph{bandwidth cost}, namely, the amount of data transferred during conversion.
For the merge regime, Maturana and Rashmi established a tight lower bound and constructed optimal MDS convertible codes achieving this bound~\cite{Maturana2023}.
Recent work has extended the study of the merge regime to more general code settings and constructions
~\cite{RKSVK25,Ge2024MDSGC,Ge2026Polynomial,Wang2026GeneralizedMerge,vajhacooperative}.
In contrast, the bandwidth problem in the split regime remains less understood.
Maturana and Rashmi derived an information-flow lower bound and proposed  explicit constructions based on the piggybacking framework~\cite{Maturana2022BandwidthCO}.
Their results characterize the optimal bandwidth when $r^F\ge k^F$ or $r^I\le r^F$.
For the remaining range, their stronger bound relies on an additional conjecture.
More recently, Wang and Hu~\cite{wang2026lower} and Singhvi \emph{et al.}~\cite{singhvi2025tight} improved the bandwidth lower bounds using algebraic and information-theoretic approaches, respectively.
Both works established tight bounds when $r^F\ge k^F$ or $r^I\le k^F$.
However, the tightness of the lower bounds for $r^F<k^F<r^I$ remained open.

In this paper, we resolve the remaining regime $r^F<k^F<r^I$ for stable MDS convertible codes with linear conversion procedures.
We strengthen the lower bound in this regime and provide explicit constructions attaining the resulting bounds.
Our main contributions are summarized as follows.
\begin{itemize}
    \item \textbf{An improved  lower bound for $\rf<\kf<\ri$.}
    Building on coordinate projections of the associated column spaces of the initial and final codes, we derive a family of rank inequalities for the read bandwidth from unchanged and retired nodes.
    These inequalities include the constraints used in ~\cite{wang2026lower} as special cases.
    In particular, when $\rf<\kf$ and $\ri\le\kf$, all additional inequalities are redundant, and the resulting linear program reduces to the two constraints used in~\cite{wang2026lower}.
    Hence, it recovers the previous lower bound in this regime, which is tight as it is matched by the constructions of Maturana and Rashmi~\cite{Maturana2022BandwidthCO}.
    For the remaining range $\rf<\kf<\ri$, the new inequalities impose additional constraints that strengthen the lower bound.
    By solving the resulting linear program, we obtain an improved  lower bound on the read bandwidth for the remaining range $\rf<\kf<\ri$.
    The bound consists of three cases and is formally stated in \cref{thm:improved-lower-bound}.
    \item \textbf{ Explicit constructions for the three cases.}
    We provide three explicit constructions corresponding to the three cases of the lower bound.
    For Cases~I and~II, we select suitable rows of the final parity matrix, referred to as \emph{anchor rows}, and use an invertible transformation to map the corresponding anchor-row matrix to a Cauchy submatrix of the initial parity matrix.
    This simultaneously yields the required column-space relation for conversion and ensures the MDS property.
    For Case~III, we construct a two-layer auxiliary MDS array code: the first layer is based on structured polynomial evaluations, while the second layer is based on a Reed--Solomon code.
    This realizes the required column-space inclusion while preserving the MDS property.
    \item \textbf{Tightness in the stable linear conversion model.}
    The three constructions attain the corresponding lower bounds whenever the required integrality conditions are satisfied.
    These conditions can always be met by choosing the subpacketization $\ell$ appropriately, with $\ell$ polynomial in the initial code length $n^I$.
    Cases~I and~II can be realized over fields of size $q=O\bigl((\ki+\ri)\ell\bigr),$
    while for Case~III it suffices to take $q=O\!\left(((\ki+\ri)\ell)^{L_0}\right)$ for some absolute constant $L_0$.
    Hence, our lower bound is tight throughout the range $\rf<\kf<\ri$ within the stable linear conversion model.
    Since the write bandwidth is fixed for stable conversion, this also determines the minimum total conversion bandwidth in this regime.
\end{itemize}

The remainder of this paper is organized as follows.
Section~\ref{sec:preliminaries} introduces the necessary definitions and notation.
Section~\ref{sec:bounds} derives improved lower bounds on the read bandwidth for the regime $\rf<\kf<\ri$.
Section~\ref{sec:constructions} presents explicit constructions attaining these lower bounds and analyzes their field-size and subpacketization requirements.
Finally, Section~\ref{sec:conclusion} concludes the paper.

\section{Preliminaries}
\label{sec:preliminaries}
In this section, we introduce the basic definitions and notation used throughout the paper.

For a positive integer $a$, let $[a]:=\{0,1,\ldots,a-1\}.$
For a matrix $A\in\F_q^{r\times s}$, we index its rows by $[r]$ and its columns by $[s]$.
For $\mathcal S\subseteq[r]$ and $\mathcal T\subseteq[s]$, let $A[\mathcal S;\mathcal T]$ denote the submatrix formed by the rows indexed by $\mathcal S$ and the columns indexed by $\mathcal T$, both taken in increasing order unless otherwise specified.
In particular, $A[\mathcal S;:]$ and $A[:;\mathcal T]$ denote the submatrices formed by the rows in $\mathcal S$ and the columns in $\mathcal T$, respectively.
The same notation is used for block matrices, in which case $\mathcal S$ and $\mathcal T$ index block rows and block columns.
For a positive integer $s$ and matrices
$A_1,\ldots,A_s$, let $\diag(A_1,\ldots,A_s)$
denote the block diagonal matrix with diagonal blocks
$A_1,\ldots,A_s$.

Let $\F_q$ be a finite field with $q$ elements.
An $[n,k,\ell]$ \emph{array code} $\mathcal C$ over $\F_q$ is a $k\ell$-dimensional subspace of $\F_q^{n\ell}$.
A codeword $\mathbf c=(\mathbf c_0,\ldots,\mathbf c_{n-1})\in \mathcal C$ consists of $n$ symbols $\mathbf c_i =(c_{i,0},\ldots,c_{i,\ell-1})\in\F_q^\ell$ for $i\in[n]$, and each scalar $c_{i,j}$ is referred to as a \emph{subsymbol}.
An $[n,k,\ell]$ array code is MDS if any $k$ out of its $n$ symbols suffice to recover the entire codeword.
A generator matrix $G$ of $\mathcal C$ has size $k\ell\times n\ell$, and its rows form a basis of $\mathcal{C}$.
For a message vector $\mathbf m\in\F_q^{k\ell}$, the encoding function of $\mathcal C$ is $\mathcal C(\mathbf m)=\mathbf mG$.

We next recall the definition of MDS convertible codes in the split regime.
\begin{definition}[MDS Convertible Codes~\cite{maturana2022convertible}]\label{def_convertible_code}
    Let $\lambda\ge2$ be an integer.
    An $[n^{I},k^{I}=\lambda\kf;n^{F},k^{F};\ell]$  MDS convertible code over  $\F_q$ consists of
    \begin{itemize}
        \item A pair of   codes $(\ci,\cf)$ where  $\ci$ is an initial  MDS array code with parameter $[\nii,\ki,\ell]$ and  $\cf$ a final  MDS array code with parameter $[\nf,\kf,\ell]$.
        \item A conversion procedure $\mathcal{T}$ with input $\{\ci(\mathbf{m}):\mathbf{m}=(\mathbf{m}_0,\cdots,\mathbf{m}_{\lambda-1})\}$  and output $\{\cf(\mathbf{m}_i):i\in [\lambda]\}$ for all  $\mathbf{m}_i\in \F_q^{\kf\ell}$.
    \end{itemize}
\end{definition}
Throughout the paper, let $\ri=\nii-\ki$ and $\rf=\nf-\kf$ denote the numbers of initial and final parity symbols, respectively.
Throughout this paper, we consider conversion procedures that read subsymbols directly from the initial nodes. A conversion procedure is called linear if each subsymbol of a new symbol is computed as an $\F_q$-linear combination of the subsymbols read from the initial codeword.
This focus is without loss of generality for practical purposes, as to the best of our knowledge, virtually all existing
constructions of convertible codes employ linear conversion procedures.

Each symbol involved in the conversion belongs to one of three categories:
(1) Unchanged symbols, which appear in both the initial and final codewords;
%For each $i\in[\lambda]$, let  $U_i$ denote the set of indices of initial codeword symbols  that remain unchanged in  the $i$-th finial codeword.  
(2) Retired symbols, which appear only in the initial codeword;
%Let $R$ be the set of indices of retired symbols in the initial codeword.  
(3) New symbols, which appear only in the final codewords.
For $i\in [\lambda]$, let  $N_i$ denote the set of indices of new symbols in the $i$-th final codeword.

For each initial symbol $i\in[\nii]$, let $ D_i\subseteq[\ell]$ be the set of subsymbol coordinates read from the $i$-th initial symbol during conversion, and define $\beta_i=|D_i|.$ The read, write, and total bandwidth costs are, respectively,
\begin{itemize}
    \item The  read bandwidth cost  is  $R=\sum\limits_{i=0}^{\nii-1}\beta_i$. 
    \item The write bandwidth cost is   $W=(\sum\limits_{i=0}^{\lambda-1}|N_i|)\ell$.
    \item The total bandwidth cost is $R+W$.
\end{itemize}

A stable conversion maximizes the number of unchanged symbols and hence minimizes the write bandwidth for the given code parameters.
\begin{definition}[Stable convertible code~\cite{maturana2022convertible}]
\label{def_stable_convertible}
An $[n^I,k^I;\,n^F,k^F;\ell]$ MDS convertible code is called \emph{stable} if it uses the maximum number of unchanged symbols among all MDS convertible codes with the same parameters.
\end{definition}
For an $ [n^I,k^I=\lambda k^F;\,n^F,k^F;\ell]$ MDS convertible code in the split regime, at most $\ki$ initial symbols can remain unchanged, and this bound can be achieved by straightforward re-encoding~\cite{wang2026lower}.
As in  previous work~\cite{Maturana2022BandwidthCO,singhvi2025tight,wang2026lower}, we focus on stable convertible codes.
In this case, the conversion has exactly $\ki$ unchanged symbols, so each final codeword contains $\kf$ unchanged symbols and $\rf$ new symbols. 
Hence, $W=\lambda\rf\ell.$
Therefore, minimizing the total conversion bandwidth is equivalent to minimizing the read bandwidth $R$.

The following lemma establishes a systematic representation for generator matrices of the initial and final codes.
\begin{lemma}[Systematic representation]
\label{lem:systematic-representation}
Consider a stable $[n^I,\ki=\lambda\kf;\,n^F,\kf;\ell]$ linear MDS convertible code in the split regime.
After suitable node permutations and invertible changes of message coordinates, we may take the first $\ki$ symbols of the initial code to be the unchanged systematic symbols, partitioned into $\lambda$ groups of $\kf$ symbols corresponding to the $\lambda$ final codewords.
The generator matrices of the initial code and the final code for each group then have the systematic forms
\[
    G^I=[I_{\ki\ell}\mid B],
    \qquad
    G_t^F=[I_{\kf\ell}\mid C_t],
    \quad t\in[\lambda].
\]
Write $B=(B_{i,j})_{i\in[\ki],\,j\in[\ri]},$ and $C_t=(C_{t;(i,j)})_{i\in[\kf],\,j\in[\rf]},$ where all blocks have size $\ell\times\ell$.
Moreover, matrices $B$ and $C_t$ are block superregular:
every square block submatrix is nonsingular.
\end{lemma}
\begin{proof}
Let $G_1$ and $G_2$ be the original initial and final generator matrices, respectively.
For each $t\in[\lambda]$, let $U_t^F$ be the set of the $\kf$ unchanged node positions in the $t$-th final codeword, and let $\overline U_t^F=[\nf]\setminus U_t^F$.
Let $P_t$ be the block permutation matrix that places the positions in $U_t^F$ first, followed by those in $\overline U_t^F$, each in increasing order.
Write
\[
    G_2P_t=[A_t\mid H_t],
    \qquad
    \text{where} \ A_t=G_2[:;U_t^F],\ \text{and}\    H_t=G_2[:;\overline U_t^F].
\]
Since the final code is MDS, each $A_t$ is nonsingular.

Let $U^I$ be the set of the $\ki$ unchanged initial node positions.
Order these positions by their corresponding final codewords, consistently with the chosen orderings of $U_0^F,\ldots,U_{\lambda-1}^F$.
Let $P$ be the block permutation matrix placing $U^I$ first in this order, followed by its complement $\overline U^I=[\nii]\setminus U^I$ in increasing order.
Set $A=\operatorname{diag}(A_0,\ldots,A_{\lambda-1}).$
By the unchanged-node correspondence,
\[
    \mathbf mG_1[:;U^I]=\mathbf mA
    \qquad
    \text{for every }
    \mathbf m=(\mathbf m_0,\ldots,\mathbf m_{\lambda-1})
    \in\F_q^{\ki\ell}.
\]
Consequently, $G_1[:;U^I]=A.$
Applying the corresponding invertible row transformations
gives
\[
\begin{aligned}
    G^I
    &=A^{-1}G_1P
      =[I_{\ki\ell}\mid B],
    & B&=A^{-1}G_1[:;\overline U^I],\\
    G_t^F
    &=A_t^{-1}G_2P_t
      =[I_{\kf\ell}\mid C_t],
    & C_t&=A_t^{-1}H_t,
    \quad t\in[\lambda].
\end{aligned}
\]
The row transformations correspond to the changes of message coordinates $\widehat{\mathbf m}_t=\mathbf m_tA_t$ and $\widehat{\mathbf m}=\mathbf mA$.
Together with the consistent node permutations, they preserve the unchanged-node correspondence and both bandwidth costs.
Moreover, invertible row transformations and node permutations preserve the MDS property, so $B$ and each $C_t$ are block superregular.
This can be viewed as a block-matrix extension of the superregular matrices defined in Section II.B of~\cite{maturana2022convertible}.
\end{proof}
\begin{remark}
The matrices $C_t$ need not coincide, since different final
codewords may have different unchanged node positions.
However, they are not chosen independently: each is obtained
from the same final generator matrix $G_2$ through
$C_t=A_t^{-1}H_t$.
In~\cite{wang2026lower}, although the matrices $C_t$ are
assumed to coincide, the lower-bound arguments remain valid
when they differ, since these arguments use only the block-superregular property of each $C_t$.
\end{remark}

Define $x=\sum_{i=0}^{\ki-1}\beta_i,$ and  $y=\sum_{i=\ki}^{n^I-1}\beta_i.$
Thus, $x$ and $y$ are the total amounts of data read from the unchanged and retired symbols, respectively, and we have the read bandwidth cost $ R=x+y.$

We next recall a characterization of linear conversion established in ~\cite{wang2026lower}.
For $E\subseteq[\ell]$, let $\overline E=[\ell]\setminus E$ be its complementary set.
For each $t\in[\lambda]$, define
$    C^{(t)}
    =\left(
        C_{t;(i,j)}[\overline D_{t\kf+i};:]
      \right)_{i\in[\kf],\,j\in[\rf]}.$  
and $\widetilde C=\operatorname{diag}\bigl(C^{(0)},C^{(1)},\ldots,C^{(\lambda-1)}\bigr).$ 
Let $\widetilde B=\left(B_{i,j}[\overline D_i;D_{\ki+j}]\right)_{i\in[\ki],j\in[\ri]}$.

\begin{lemma}[Linear-conversion characterization]
\label{lem:conversion-characterization}
Let $(\mathcal C^I,\mathcal C^F)$ be a stable $[n^I,\ki=\lambda\kf;\,n^F,\kf;\ell]$ MDS convertible code
with systematic representations $G^I$ and $G_t^F$, $t\in[\lambda]$, as described in Lemma~\ref{lem:systematic-representation}.
For any $D_i\subseteq[\ell]$, $i\in[n^I]$, there exists a linear stable conversion from $C^I$ to $C^F$ using only the subsymbols specified by the sets $D_i$ if and only if
\begin{equation}
    \col(\widetilde C)
    \subseteq
    \col(\widetilde B).
    \label{eq:conversion-characterization}
\end{equation}
Moreover, such a conversion can be performed with read bandwidth $x+\rank(\widetilde B).$
In particular, if the linear conversion procedure minimizes the read bandwidth, then $\widetilde B$ may be assumed to have full column rank.
\end{lemma}
The necessity of \eqref{eq:conversion-characterization} and
the full-column-rank property follow by the same arguments
as in \cite[Lemma~2]{wang2026lower}, while the converse
follows by the argument in \cite[Remark~1]{wang2026lower}.
These arguments do not require the matrices $C_t$ to coincide.

We will use the following form of Linnik's theorem in
Section~\ref{caseIII_field}.

\begin{lemma}[Linnik's theorem~\cite{linnik1944least}]
\label{lem:linnik}
There exist absolute constants $c,L>0$ such that, for every
integer $d\ge 2$ and every integer $a$ satisfying
$1\le a<d$ and $\gcd(a,d)=1$, the least prime $p$ with
$p\equiv a\pmod d$ satisfies
\[
    p\le c d^L.
\]
\end{lemma}
% We will use the following Linnik's theorem in Section~\ref{caseIII_field}
% \begin{lemma}[Linnik's ntheorem~\cite{linnik1944least}]\label{Linnik's_theorem}
% There exist positive $c$ and $L$ such that, if we denote $p(a,d)$ the least prime in the arithmetic progression $a+nd$, where $n$ runs through the positive integers and a and $d$ are any given positive coprime integers with $1 ≤ a ≤ d − 1$, then
% $p(a,d)<cd^L.$
% \end{lemma}

% For a matrix $A$, let
% \[
%     \col(A)
% \]
% denote its column space.
% Throughout the remainder of the paper, we restrict attention to stable
% MDS convertible codes with linear conversion procedures in the split
% regime and assume
% \[
%     \rf<\kf<\ri,
%     \qquad
%     \ki=\lambda\kf.
% \]
% For stable conversion, exactly $\ki$ symbols are unchanged.
% Without loss of generality, we take these unchanged symbols to be the
% $\ki$ systematic symbols of the initial code, partitioned into
% $\lambda$ groups of $\kf$ symbols corresponding to the $\lambda$ final
% codewords.
% Consequently, the write bandwidth is fixed as
% \[
%     W=\lambda\rf\ell,
% \]
% and minimizing the total conversion bandwidth is equivalent to
% minimizing the read bandwidth $R$.
\section{A Family of Rank Inequalities and Improved Lower Bounds}
\label{sec:bounds}

In this section, we derive a family of rank inequalities for stable MDS convertible codes with linear conversion procedures in the split regime, and use them to obtain improved lower bounds on the read bandwidth.
% The first two members of this family recover the constraints used in
% our previous work~\cite{wang2026lower}, while the remaining inequalities
% provide additional constraints on the read bandwidth.
% We then show that these additional constraints lead to improved lower
% bounds precisely in the remaining regime
% \[
%     \rf<\kf<\ri.
% \]
For $\rf\ge\kf$, the minimum read bandwidth is already known to be
$\ki\ell$, achieved by full re-encoding ~\cite{Maturana2022BandwidthCO,wang2026lower,singhvi2025tight}.
Therefore, throughout this section, we focus on the nontrivial case $\rf<\kf.$

\subsection{A Family of Rank Inequalities}
For each $t\in[\lambda]$, define $ x_t=\sum_{j=t\kf}^{(t+1)\kf-1}\beta_j.$
Thus, $x_t$ is the total amount of data read from the $\kf$ unchanged symbols associated with the $t$-th final codeword, and $ x=\sum_{t\in[\lambda]}x_t.$

Consider a linear conversion procedure that minimizes the read bandwidth.
By Lemma~\ref{lem:conversion-characterization}, $\col(\widetilde C)\subseteq\col(\widetilde B),$ and $\widetilde B$ has full column rank.
Since $\widetilde B$ has $y=\sum_{j=\ki}^{n^I-1}\beta_j$ columns, it follows that $\rank(\widetilde B)=y.$

We first derive a lower bound on the rank of each diagonal block $C^{(t)}$ of $\widetilde C$.
For any $U\subseteq[\kf]$ with $|U|=\rf$, the $\rf\times\rf$ block submatrix $C_t[U;:]$ is nonsingular since $C_t$ is block superregular.
Deleting the rows indexed by $D_{t\kf+i}$ within each block row $i\in U$ decreases the rank by at most
$\sum_{i\in U}\beta_{t\kf+i}$.
The resulting matrix is a submatrix of $C^{(t)}$.
Therefore,

\[
    \rank(C^{(t)})
    \ge
    \rf\ell-\sum_{i\in U}\beta_{t\kf+i}.
\]
Averaging over all $\rf$-subsets $U\subseteq[\kf]$ gives $\rank(C^{(t)})\ge\rf\ell-\frac{\rf}{\kf}x_t.$
Since $\widetilde C=\operatorname{diag}\bigl(C^{(0)},\ldots,C^{(\lambda-1)}\bigr),$
summing  over $t\in[\lambda]$ yields
\begin{equation}\label{eq:rank-Ctilde}
    \rank(\widetilde C)\ge\lambda\rf\ell-\frac{\rf}{\kf}x.
\end{equation}

We next derive additional rank constraints by applying coordinate projections to the column spaces of the initial and final codes.
\begin{lemma}
\label{thm:rank-family}
For  every $s\in[\lambda]$, define $q_s:=\min\{s\kf,\ri\}.$
Then the read amounts $x$ and $y$ satisfy
\begin{equation}
    \frac{q_s+(\lambda-s)\rf}{\ki}x+\left(1-\frac{q_s}{\ri}\right)y\ge(\lambda-s)\rf\ell.
    \label{eq:rank-family}
\end{equation}
\end{lemma}

\begin{proof}
We first consider $1\le s\le\lambda-1$.
Fix a subset $S\subseteq[\lambda]$ with $|S|=s$, let $\mathcal I_S=\bigcup_{t\in S}\{t\kf,t\kf+1,\ldots,(t+1)\kf-1\},$ and $ x_S:=\sum_{t\in S}x_t.$
Let $\widetilde B_S$ and $\widetilde C_S$  be obtained from $\widetilde B$ and $\widetilde C$, respectively, by retaining the block rows indexed by $\mathcal I_S$.
Since $\widetilde C$ is block diagonal,
\begin{equation}
    \rank(\widetilde C_S)
    =
    \sum_{t\in S}\rank(C^{(t)}).
    \label{eq:rank-CS}
\end{equation}
Let $\pi_S$ denote the coordinate projection that retains the block rows indexed by $\mathcal I_S$, then we have 
$\col(\widetilde B_S)=\pi_S\bigl(\col(\widetilde B)\bigr)$ and  $\col(\widetilde C_S)=\pi_S\bigl(\col(\widetilde C)\bigr).$ 
Since $\col(\widetilde C)\subseteq\col(\widetilde B),$ we have $\left(\col(\widetilde C)\cap\ker\pi_S\right)\subseteq\left(\col(\widetilde B)\cap\ker\pi_S\right).$
By rank-nullity,
\begin{equation}
\begin{aligned}
    \rank(\widetilde B_S)-\rank(\widetilde C_S)=&\rank(\widetilde B)-\rank(\widetilde C) \\
    &-\Bigl(\dim\left(\col(\widetilde B)\cap\ker\pi_S\right)-\dim\left(\col(\widetilde C)\cap\ker\pi_S\right)\Bigr) \\
    \le&
    \rank(\widetilde B)-\rank(\widetilde C)\\
    =&
    y-\rank(\widetilde C).
\end{aligned}
\label{eq:projected-rank}
\end{equation}
We now lower bound $\rank(\widetilde B_S)$.
Let $ H\subseteq\mathcal I_S$, $J\subseteq[\ri]$, and $|H|=|J|=q_s.$
Since $B$ is block superregular, the $q_s\times q_s$ block submatrix $B[H;J]$ is nonsingular.
Restricting each block column indexed by $j\in J$ to the downloaded coordinates $D_{\ki+j}$ preserves the linear independence of the selected columns.
Deleting the rows indexed by $D_i$, $i\in H$, can decrease the rank by at most $\sum_{i\in H}\beta_i$.
Therefore,
\begin{equation}
    \rank(\widetilde B_S)
    \ge
    \sum_{j\in J}\beta_{\ki+j}
    -
    \sum_{i\in H}\beta_i.
    \label{eq:rank-BS-HJ}
\end{equation}
Summing \eqref{eq:rank-BS-HJ} over all $q_s$-subsets
$J\subseteq[\ri]$ and all $q_s$-subsets
$H\subseteq\mathcal I_S$ gives
$$\binom{s\kf}{q_s}\binom{\ri}{q_s}\rank(\widetilde B_S)\ge\binom{s\kf}{q_s}\binom{\ri-1}{q_s-1}\sum_{j\in[\ri]}\beta_{\ki+j}-\binom{s\kf-1}{q_s-1}\binom{\ri}{q_s}x_S.$$
By averaging all $q_s$-subsets $J\subseteq[\ri]$ and all $q_s$-subsets, we have
\begin{equation}
    \rank(\widetilde B_S)
    \ge
    \frac{q_s}{\ri}y
    -
    \frac{q_s}{s\kf}x_S.
    \label{eq:rank-BS}
\end{equation}

Averaging \eqref{eq:projected-rank}  and  \eqref{eq:rank-BS} over all $S\subseteq[\lambda]$ with $|S|=s$, we have 
\begin{align}
    \frac{1}{\binom{\lambda}{s}}
    \sum_{\substack{S\subseteq[\lambda]\\|S|=s}}
    \rank(\widetilde B_S)-\frac{s}{\lambda}\rank(\widetilde C)\le y-\rank(\widetilde C) \label{eq:avg-5}\\
    \frac{1}{\binom{\lambda}{s}}
    \sum_{\substack{S\subseteq[\lambda]\\|S|=s}}
    \rank(\widetilde B_S)
    \ge
    \frac{q_s}{\ri}y-\frac{q_s}{\ki}x.
    \label{eq:avg-BS}
\end{align}
Combining \eqref{eq:avg-5} and \eqref{eq:avg-BS}, we obtain
\begin{equation}
    \frac{q_s}{\ki}x
    +
    \left(1-\frac{q_s}{\ri}\right)y
    \ge
    \frac{\lambda-s}{\lambda}\rank(\widetilde C).
    \label{eq:rank-family-intermediate}
\end{equation}
Using (\ref{eq:rank-Ctilde}) and $\ki=\lambda\kf$ gives
\[
    \frac{q_s+(\lambda-s)\rf}{\ki}x
    +
    \left(1-\frac{q_s}{\ri}\right)y
    \ge
    (\lambda-s)\rf\ell,
\]
which proves \eqref{eq:rank-family} for
$1\le s\le\lambda-1$.

For $s=0$, we have $q_0=0$.
Since $\rank(\widetilde C)\le\rank(\widetilde B)=y,$ together with (\ref{eq:rank-Ctilde}), we obtain $\frac{\rf}{\kf}x+y\ge\lambda\rf\ell,$ which is equivalent to \eqref{eq:rank-family} for $s=0$.
\end{proof}
\begin{remark}
The proof of Lemma~\ref{thm:rank-family} uses only the
block-superregular property of each $C_t$ and does not
require the final codes to coincide or be equivalent.
Hence, the rank inequalities also hold for stable linear
conversions from one $[n^I,\ki=\lambda\kf,\ell]$ MDS array
code to $\lambda$  different $[n^F,\kf,\ell]$ MDS array codes.
\end{remark}

Lemma~\ref{thm:rank-family} gives a family of rank inequalities satisfied by the optimal linear conversion, which includes the constraints used in~\cite{wang2026lower} as special cases.

When $\ri\le\kf$ and $\rf<\kf$, the constraint corresponding to $s=0$ recovers the first constraint used in the proof of Theorem 2 of ~\cite{wang2026lower}. 
Moreover, for every $1\le s\le\lambda-1$, we have $q_s=\ri$, and \eqref{eq:rank-family} reduces to $x\ge\frac{\ki(\lambda-s)\rf}{\ri+(\lambda-s)\rf}\ell.$
Since the right-hand side decreases with $s$, the constraint corresponding to  $s=1$ is the strongest among all $s\ge1$. 
Hence, all inequalities with $s\ge2$ are redundant, and the resulting LP reduces exactly to the two constraints used in Theorem~2 of~\cite{wang2026lower}.
Moreover, this lower bound is tight in the regime $\ri\le\kf$ and $\rf<\kf$, as it is matched by the construction
of Maturana and Rashmi~\cite{Maturana2022BandwidthCO}.

When $\rf<\kf<\ri$, the constraints corresponding to $s=0$ and $s=1$ coincide with the two constraints underlying Theorem~3 of~\cite{wang2026lower}, while  the constraints with $s\ge2$ provide additional constraints that can strengthen the resulting lower bound.
\subsection{Improved Lower Bounds for $\rf<\kf<\ri$}
\label{subsec:improved-bounds}

We now use the family of rank inequalities in Lemma~\ref{thm:rank-family} to derive improved lower bounds for the
remaining regime $ \rf<\kf<\ri.$
% Since $R=x+y$, it suffices to minimize $x+y$ subject to
% \eqref{eq:rank-family} and
% \[
%     0\le x\le \ki\ell,
%     \qquad
%     0\le y\le \ri\ell.
% \]

% We first state the resulting lower bound.

\begin{thm}
\label{thm:improved-lower-bound}
For every stable MDS $[\nii,\ki=\lambda\kf;\nf,\kf;\ell]$ convertible code with a linear conversion procedure satisfying $\rf<\kf<\ri$, define $m=\left\lfloor\frac{\ri-1}{\kf}\right\rfloor,$ $\delta=\ri-m\kf,$ and $ h=\lambda-m$. 
Further, let $ \Delta:=\ri(\kf-\rf)+\ki\rf,$ $T:=\ri+(h-1)\rf,$ and define $ \Gamma=\frac{\ri(m\kf+h\rf)}{\ki\delta}.$
Then the read bandwidth cost satisfies
    $$R\ge\begin{cases}
\lambda\, \rf \, \ell, & \text{if } \ \ki \le \ri,\\[2mm]
 \frac{(\ki)^2  \rf}{\Delta} \, \ell, & \text{if } \ \ki > \ri,\Gamma\ge  1,\\[2mm]
 \displaystyle
    \left(
    \frac{\ki(h-1)\rf}{T}
    +
    \frac{\ri\rf(m\kf+h\delta)}
    {\delta T}
    \right)\ell,
 & \text{if }\  \ki > \ri,\Gamma<1 .
\end{cases}$$
\end{thm}

\begin{proof}
It suffices to consider a linear conversion procedure that minimizes the read bandwidth.
By Lemma~\ref{thm:rank-family}, for every $s\in[\lambda]$, we have 
\begin{equation} 
    \frac{q_s+(\lambda-s)\rf}{\ki}x
    +
    \left(1-\frac{q_s}{\ri}\right)y
    \ge
    (\lambda-s)\rf\ell,
    \label{eq:family-for-lp}
\end{equation}
where $q_s=\min\{s\kf,\ri\}$.
We consider the linear program
\begin{equation}\label{eq:split-lp}
    \begin{array}{rl}
    \text{minimize} & R = x+y \\
    \text{subject to} & \eqref{eq:family-for-lp}, \\
                      &  0 \le x \le \ki\ell,\  0 \le y \le \ri\ell.
    \end{array}
    \end{equation}
The optimal value of this LP provides a lower bound on the read bandwidth.
% It suffices to consider a minimum-read linear conversion procedure.
% Let $x$ and $y$ denote the total amounts of data read from the
% unchanged and retired symbols, respectively, so that
% \[
%     R=x+y.
% \]

We first consider the case  $\ki\le\ri$.
Taking $s=0$ in \eqref{eq:family-for-lp} gives $\frac{\rf}{\kf}x+y \ge\lambda\rf\ell.$
 Since $\rf<\kf$ and $x\ge0$, we have
 \[
     R=x+y
     =
     \left(\frac{\rf}{\kf}x+y\right)
     +
     \left(1-\frac{\rf}{\kf}\right)x
     \ge
     \lambda\rf\ell.
 \]
 Moreover, the point $(x,y)=(0,\lambda\rf\ell)$ is feasible for  (\ref{eq:split-lp}).
More specifically, for $s=0$, equality holds in \eqref{eq:family-for-lp}.
For $1\le s\le\lambda-1$,  we have \(q_s=s\kf\), and substituting this point into \eqref{eq:family-for-lp}
gives
\[
    \left(1-\frac{s\kf}{\ri}\right)
    \lambda\rf\ell
    \ge
    (\lambda-s)\rf\ell.
\]
This inequality is equivalent to \(\lambda\kf\le\ri\),which holds by assumption.
The variable bounds are also satisfied since
$0<\lambda\rf\ell<\ki\ell\le\ri\ell$.
Thus, the optimal value of \eqref{eq:split-lp} is
$\lambda\rf\ell$, proving the first case.

We next  consider $\ki>\ri$.
Since  $\kf<\ri<\ki=\lambda\kf,$  by the definition $m=\left\lfloor\frac{\ri-1}{\kf}\right\rfloor$ we have  $1\le m\le\lambda-1$ and  $ m\kf<\ri\le(m+1)\kf.$
Recall that $ \ri=m\kf+\delta$ and $h=\lambda-m.$
We now simplify the constraints in \eqref{eq:family-for-lp}.
For $0\le s\le m$, we have $q_s=s\kf$,  and \eqref{eq:family-for-lp}  becomes
\begin{equation}
    \frac{s\kf+(\lambda-s)\rf}{\ki}x +
    \frac{\ri-s\kf}{\ri}y
    \ge
    (\lambda-s)\rf\ell.
    \label{eq:family-unsaturated}
\end{equation}
Let $F_s(x,y)=\frac{s\kf+(\lambda-s)\rf}{\ki}x+\frac{\ri-s\kf}{\ri}y-(\lambda-s)\rf\ell$. 
Since $F_s(x,y)=F_0(x,y)+s(F_1(x,y)-F_0(x,y))$, all equality lines  $ F_s(x,y)=0$ where $0\le s\le m,$ pass through the common intersection of \(F_0(x,y)=0\) and \(F_1(x,y)=0\).
A direct calculation gives this point as
    $$P=(x_{\rm int},y_{\rm int})=\left(\frac{\ki\rf(\ki-\ri)}{\Delta}\ell,\frac{\ki\rf\ri}{\Delta}\ell\right).$$
For \(0\le s\le m\), write the equality line in \eqref{eq:family-unsaturated} as $ y=\phi_s(x)$. Its slope is $\gamma_s=-\frac{\ri\bigl(s\kf+(\lambda-s)\rf\bigr)}{\ki(\ri-s\kf)}.$
For \(0\le s\le m-1\),
\begin{equation*}
    \gamma_{s+1}-\gamma_s=-\frac{
        \ri\bigl(
        \lambda\kf\rf+\ri(\kf-\rf)
        \bigr)
    }{
        \ki(\ri-s\kf)
        (\ri-(s+1)\kf)
    } 
    <0.
\end{equation*}
Thus the slopes are strictly decreasing: $ \gamma_0>\gamma_1>\cdots>\gamma_m.$
Since all these lines pass through \(P\), it follows that
\begin{equation}
    \max_{0\le s\le m}\phi_s(x)
    =\begin{cases}
        \phi_m(x), & x\le x_{\rm int},\\
        \phi_0(x), & x\ge x_{\rm int}.
    \end{cases}
    \label{eq:unsaturated-envelope}
\end{equation}
For $m+1\le s\le\lambda-1$, we have $q_s=\ri$, and
\eqref{eq:family-for-lp} reduces to $ x\ge\frac{\ki(\lambda-s)\rf}{\ri+(\lambda-s)\rf}\ell.$
The right-hand side is decreasing in \(s\). 
Hence, if the range \(m+1\le s\le\lambda-1\) is nonempty, its strongest constraint is obtained at \(s=m+1\). 
Together with \(x\ge0\), all saturated constraints are equivalent to $x\ge x_{\rm v}=\frac{\ki(h-1)\rf}{T}\ell,$ When \(m=\lambda-1\), we have \(h=1\), and hence $x_{\rm v}=0$.
A direct calculation yields
\begin{equation}
    x_{\rm int}-x_{\rm v}
    =
    \frac{
        \ki\ri\rf(\kf-\delta)
    }{
        \Delta T
    }\ell
    \ge0,
    \label{eq:x-int-v}
\end{equation}
where the last inequality follows from \(1\le\delta\le\kf\).
Consequently, the constraints in \eqref{eq:split-lp}
are equivalent to $x_{\rm v}\le x\le\ki\ell,$ $0\le y\le\ri\ell,$
together with
\[
    y\ge
    \begin{cases}
        \phi_m(x),
            &x_{\rm v}\le x\le x_{\rm int},\\
        \phi_0(x),
            &x\ge x_{\rm int}.
    \end{cases}
\]
The two relevant slopes are
\begin{equation}
    \gamma_m =-\frac{\ri(m\kf+h\rf)}{\ki\delta}=-\Gamma,\qquad\gamma_0= -\frac{\lambda\rf}{\ki}=-\frac{\rf}{\kf}.
\label{eq:two-slopes}
\end{equation}
Since \(\rf<\kf\), we have $-1<\gamma_0<0.$ 
The level lines of the objective $x+y$ have slope $-1$.

We now distinguish two cases.
If \(\Gamma\ge1\), then $\gamma_m\le -1<\gamma_0.$
Moreover, we have $\frac{x_{\rm int}}{\ki\ell}=\frac{\rf(\ki-\ri)}{\Delta}<1,$ and $\frac{y_{\rm int}}{\ri\ell}=\frac{\ki\rf}{\Delta}<1.$
Thus, the point $P=(x_{\rm int},y_{\rm int})$ lies in the feasible region.
Hence, as illustrated in Figure~\ref{fig:lp-gamma-geq-one}, the minimum of $x+y$ is attained at $ P=(x_{\rm int},y_{\rm int})=\left(
        \frac{\ki\rf(\ki-\ri)}{\Delta}\ell,\,
        \frac{\ki\rf\ri}{\Delta}\ell
    \right).$

Therefore,
\[
\begin{aligned}
    R
    &\ge x_{\rm int}+y_{\rm int}\\
    &=
    \frac{\ki\rf(\ki-\ri)+\ki\rf\ri}{\Delta}\ell\\
    &=
    \frac{(\ki)^2\rf}{\Delta}\ell,
\end{aligned}
\]
which proves the second case.

If \(\Gamma<1\), then  $-1<\gamma_m<\gamma_0<0.$
We first note that $\frac{x_{\rm v}}{\ki\ell}=\frac{(h-1)\rf}{T}<1$.
Moreover, $\Gamma<1$ implies $\delta>\rf$. 
Indeed, if $\delta\le\rf$, then
\[
    \ri(m\kf+h\rf)
    \ge
    (m\kf+\delta)(m\kf+h\delta)
    >
    \lambda\kf\delta
    =
    \ki\delta,
\]
contradicting $\Gamma<1$. Hence $\delta T-\rf(m\kf+h\delta) =\ri(\delta-\rf)>0,$
and therefore
$\frac{y_{\rm v}}{\ri\ell}
    =
    \frac{\rf(m\kf+h\delta)}{\delta T}
    <1.$
Thus, the point $Q=(x_{\rm v},y_{\rm v})$ lies in the feasible region.
Hence, as illustrated in Figure~\ref{fig:lp-gamma-lt-one},  the minimum is attained at the leftmost feasible point $$Q=(x_{\rm v},y_{\rm v})=\left(\frac{\ki(h-1)\rf}{T}\ell,\frac{\ri\rf(m\kf+h\delta)}{\delta T}\ell\right).$$
It follows that 
\[
    R
    \ge
    x_{\rm v}+y_{\rm v}
    =
    \left(
    \frac{\ki(h-1)\rf}{T}
    +
    \frac{\ri\rf(m\kf+h\delta)}
         {\delta T}
    \right)\ell,
\] 
which proves the third case.
\end{proof}
\begin{figure*}[t]
\centering

%==================== subfigure (a) ====================
\begin{subfigure}[t]{0.48\textwidth}
\centering
\resizebox{\linewidth}{!}{%
\begin{tikzpicture}[
    x=5.5cm,
    y=5.5cm,
    >=Latex,
    every node/.style={font=\scriptsize}
]

% Parameters: (lambda,k^F,r^F,r^I,ell)=(4,6,1,8,1).
% k^I=24, Gamma=3/2.
% Actual coordinates:
% L_0: y=4-x/6; L_m: y=12-3x/2.
% P=(6,3), Q=(24/5,24/5), R^star=9.
% Drawing coordinates: X=x/24, Y=y/8.

% Identical bounding boxes for both panels
\path[use as bounding box]
    (-0.23,-0.12) rectangle (1.16,1.15);

\coordinate (P) at (0.25,0.375);
\coordinate (Q) at (0.2,0.6);

% Feasible region
\fill[
    pattern={Lines[angle=45,distance=4pt,line width=0.3pt]},
    pattern color=blue!60
]
    (Q) -- (0.2,1) -- (1,1) --
    (1,0) -- (P) -- cycle;

% Upper, right, and vertical boundaries
\draw[thin,gray] (0.2,1) -- (1,1) -- (1,0);
\draw[thick] (Q) -- (0.2,1);

% Guides
\draw[densely dotted] (0,1) -- (0.2,1);
\draw[densely dotted] (0.2,0) -- (Q);
\draw[densely dotted] (0.25,0) -- (P);

% Constraint and objective lines
\begin{scope}
    \clip (0,0) rectangle (1,1);

    % L_0
    \draw[thick] (0,0.5) -- (1,0);

    % L_m
    \draw[thick] (0,1.5) -- ({1/3},0);

    % x+y=R^star
    \draw[dashed,thick] (0,1.125) -- (0.375,0);
\end{scope}

% Axes
\draw[->] (0,0) -- (1.07,0) node[right] {$x$};
\draw[->] (0,0) -- (0,1.07) node[above] {$y$};

% Axis labels
\node[left] at (0,1) {$r^I\ell$};
\node[left] at (0,0.5) {$\lambda r^F\ell$};
\node[below left] at (0.2,0) {$x_{\rm v}$};
\node[below right] at (0.25,0) {$x_{\rm int}$};
\node[below] at (1,0) {$k^I\ell$};

% Line labels: no white backgrounds
\node[above right] at (0.73,0.135) {$\mathcal L_0$};
\node[right] at (0.13,0.915) {$\mathcal L_m$};
\node[right] at (0.35,0.075) {$x+y=R^\star$};

% Optimal point
\fill[blue!70!black] (P) circle (1.7pt);
\node[above right,text=blue!70!black] at (P) {$P$};

\end{tikzpicture}%
}
\caption{$\Gamma\ge1$.}
\label{fig:lp-gamma-geq-one}
\end{subfigure}
\hfill
%==================== subfigure (b) ====================
\begin{subfigure}[t]{0.48\textwidth}
\centering
\resizebox{\linewidth}{!}{%
\begin{tikzpicture}[
    x=0.6875cm,
    y=0.6875cm,
    >=Latex,
    every node/.style={font=\scriptsize}
]

% Schematic coordinates; not a numerical parameter instance.
% L_0: y=4-x/2.
% L_m: y=5-3x/4.
% P=(4,2), Q=(1.5,3.875).
% Optimal objective line: x+y=5.375.

% Same physical bounding box as the previous left panel
\path[use as bounding box]
    (-1.84,-0.96) rectangle (9.28,9.2);

\coordinate (P) at (4,2);
\coordinate (Q) at (1.5,3.875);

% Feasible region
\fill[
    pattern={Lines[angle=45,distance=4pt,line width=0.3pt]},
    pattern color=blue!60
]
    (Q) -- (1.5,8) -- (8,8) --
    (8,0) -- (P) -- cycle;

% Upper, right, and vertical boundaries
\draw[thin,gray] (1.5,8) -- (8,8) -- (8,0);
\draw[thick] (Q) -- (1.5,8);

% Guides
\draw[densely dotted] (0,8) -- (1.5,8);
\draw[densely dotted] (1.5,0) -- (Q);
\draw[densely dotted] (4,0) -- (P);

% Constraint lines
\draw[thick] (0,4) -- (8,0);
\draw[thick] (0,5) -- ({20/3},0);

% Optimal objective line through Q
\draw[dashed,thick] (0,5.375) -- (5.375,0);

% Axes
\draw[->] (0,0) -- (8.56,0) node[right] {$x$};
\draw[->] (0,0) -- (0,8.56) node[above] {$y$};

% Axis labels
\node[left] at (0,8) {$r^I\ell$};
\node[left] at (0,4) {$\lambda r^F\ell$};
\node[below] at (1.5,0) {$x_{\rm v}$};
\node[below] at (4,0) {$x_{\rm int}$};
\node[below] at (8,0) {$k^I\ell$};

% Line labels
\node[above right] at (6.2,0.9) {$\mathcal L_0$};
\node[above right] at (2.5,3.125) {$\mathcal L_m$};
\node[below left] at (3.3,2.075) {$x+y=R^\star$};

% Intersection and optimal point
\fill (P) circle (1.5pt);
\node[above right] at (P) {$P$};

\fill[blue!70!black] (Q) circle (1.8pt);
\node[above right,text=blue!70!black] at (Q) {$Q$};

\end{tikzpicture}%
}
\caption{$\Gamma<1$.}
\label{fig:lp-gamma-lt-one}
\end{subfigure}

\caption{Geometric interpretation of the linear program in \cref{thm:improved-lower-bound}.
The lines $\mathcal L_0$ and $\mathcal L_m$ are the equality boundaries corresponding to $s=0$ and $s=m$, respectively.
The dashed lines represent $x+y=R^\star$, where $R^\star$ is the optimal objective value.
}
\label{fig:lp-two-cases}
\end{figure*}
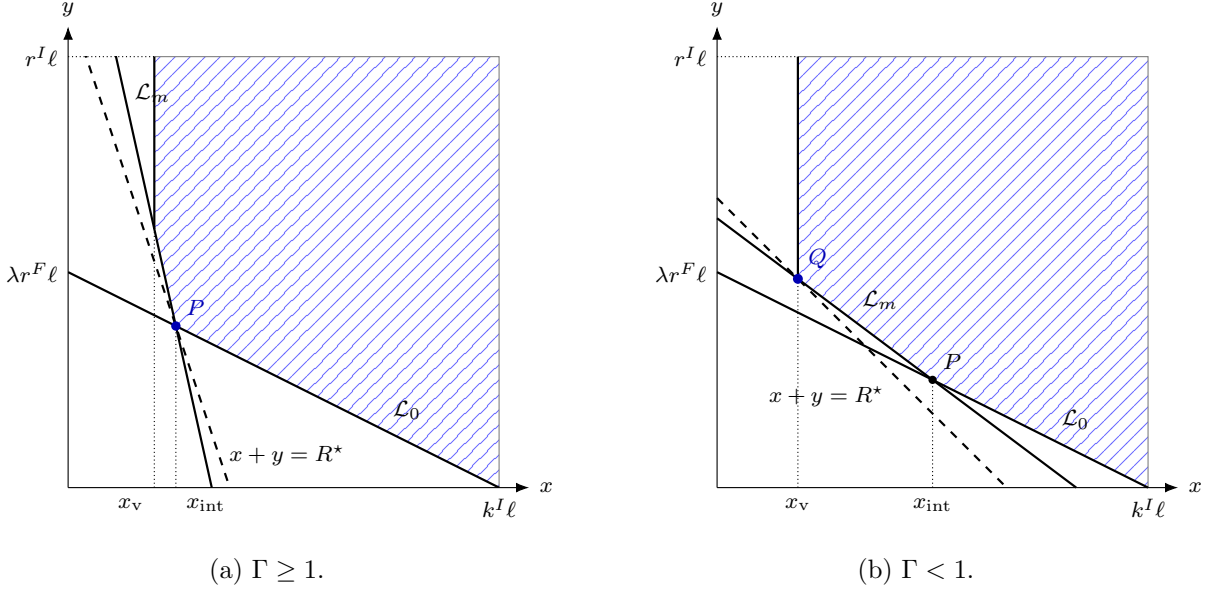
\begin{remark}
\label{rem:optimal-read-pairs}
The proof of \cref{thm:improved-lower-bound} also identifies
an optimal solution $(x^*,y^*)$ to the LP
\eqref{eq:split-lp}, given by
\[
    (x^*,y^*)
    =
    \begin{cases}
    \bigl(0,\lambda\rf\ell\bigr),
    & \ki\le\ri,\\[0.6em]
    \left(
    \dfrac{\ki\rf(\ki-\ri)}{\Delta}\ell,
    \dfrac{\ki\rf\ri}{\Delta}\ell
    \right),
    & \kf<\ri<\ki,\ \Gamma\ge1,\\[1em]
    \left(
    \dfrac{\ki(h-1)\rf}{T}\ell,
    \dfrac{\ri\rf(m\kf+h\delta)}{\delta T}\ell
    \right),
    & \kf<\ri<\ki,\ \Gamma<1.
    \end{cases}
\]
These three cases will be used in
Section~\ref{sec:constructions} to construct matching convertible codes.
\end{remark}

\section{Explicit Constructions Attaining the Lower Bounds}\label{sec:constructions}

In this section, we present explicit constructions attaining the lower bounds derived in \cref{thm:improved-lower-bound}.
For each construction, we choose the subpacketization $\ell$ so that the prescribed number of subsymbols read from each node is an integer.

In each construction, all $\lambda$ final codewords use
the same systematic generator matrix. 
Specifically, we construct
 \[
     G^I=[I_{\ki\ell}\mid \Boverall],
    \ 
     G_t^F=G^F=[I_{\kf\ell}\mid \Coverall]\  \text{for} \ t\in[\lambda],
 \]
 where $B=(\Boverall_{i,j})_{i\in[\ki],j\in[\ri]}$ and $C=(\Coverall_{i,j})_{i\in[\kf],j\in[\rf]}$ with $\Boverall_{i,j},\Coverall_{i,j}\in\F_q^{\ell\times\ell}$. 
Thus, the matrices $C_t$ in Lemma~\ref{lem:systematic-representation} all equal $\Coverall$.
For each construction, we establish the block-superregular property of $\Boverall$ and $\Coverall$, ensuring that both codes are MDS, and describe a linear conversion procedure that preserves all $\ki$ systematic symbols and attains the corresponding lower bound on the read bandwidth.
% For simplicity, we assume uniform downloads among unchanged systematic nodes and among retired parity nodes.

% =====================================================================

\subsection{Case I: $\ri\ge \ki$}
\label{subsec:case-I}

In this case, the optimal read bandwidth is $R=\lambda\rf\ell.$ 
We attain this bound by downloading subsymbols only from the retired parity nodes.
Choose $\ell$ such that $\ri\mid\lambda\rf\ell$, and set $\sigma=\frac{\lambda\rf\ell}{\ri}$, $\tau=\ell-\sigma.$
Since $\rf<\kf$ and $\ri\ge\ki=\lambda\kf$,
we have $0<\sigma<\ell$.
Each retired parity node contributes $\sigma$ subsymbols, while no systematic subsymbol is read.

\subsubsection{Construction of the final code $C^F$}
Choose a scalar Cauchy matrix $C\in\F_q^{\kf\ell\times\rf\ell}$, and partition it as  $C=\begin{bmatrix} C_{i,j} \end{bmatrix}_{i\in[\kf],j\in[\rf]}$, where $C_{i,j} \in\F_q^{\ell\times\ell}.$  
Every square block submatrix of $C$ is a square submatrix of a scalar Cauchy matrix and is therefore nonsingular.
Thus, $G^F=[I_{\kf\ell}\mid C]$ generates an MDS array code.

% To construct the initial code, we prepare  the following materials.
Define $F=I_\lambda\otimes C\in\F_q^{\ki\ell\times\ri\sigma},$ where $\ri\sigma=\lambda\rf\ell$.
%%%%%%%%有条件更新！！！！！
Then $F$ is the block-diagonal replication of the final parity transform across the $\lambda$ final groups.
For $i\in[\kf]$, let $C_i=(C_{i,j})_{j\in[\rf]}$ denote the $i$-th block row of $C$.
Let $\mathcal{S}=[\sigma]$ and $\overline{\mathcal{S}}=[\ell]\backslash[\sigma]$.
We designate the first $\sigma$ rows of each $C_i$
as anchor rows and define
$$\Aone_0=\begin{bmatrix}
    C_0[\mathcal S;:]\\
    \vdots\\
    C_{\kf-1}[\mathcal S;:]
\end{bmatrix}\in\F_q^{\kf\sigma\times\rf\ell},\ \ \Mone_0=\begin{bmatrix}
    C_0[\overline{\mathcal S};:]\\
    \vdots\\
    C_{\kf-1}[\overline{\mathcal S};:]
\end{bmatrix}\in\F_q^{\kf\tau\times\rf\ell}.$$
Since $\kf\sigma=\frac{\lambda\kf\rf\ell}{\ri} \le\rf\ell$ and every square submatrix of $C$ is nonsingular, $\Aone_0$ has full row rank.
Let $\Aone=I_\lambda\otimes \Aone_0\in\F_q^{\ki\sigma\times\ri\sigma}$ and $  \Mone=I_\lambda\otimes  \Mone_0\in\F_q^{\ki\tau\times\ri\sigma}$.
Then  $\rank\Aone=\ki\sigma.$

\subsubsection{Construction of the initial code $C^I$}
Choose a scalar Cauchy matrix $E\in\F_q^{\ki\sigma\times\ri\sigma},$ and partition it into $\ki\times\ri$ blocks $ E=(E_{i,j})_{i\in[\ki],j\in[\ri]}$, where  $E_{i,j}\in\F_q^{\sigma\times\sigma}.$
Both $\Aone$ and $E$ have full row rank.
Complete their rows to bases of $\F_q^{\ri\sigma}$, obtaining invertible matrices $\widehat A$ and $\widehat E$ whose first $\ki\sigma$ rows are
$\Aone$ and $E$, respectively.
Then $L=\widehat A^{-1}\widehat E\in\mathrm{GL}_{\ri\sigma}(\F_q)$ satisfies $\Aone L=E$.
Set $Z=\Mone L$ and $D=FL$, and partition them as
\[
    Z=(Z_{i,j})_{i\in[\ki],\,j\in[\ri]},
    \qquad
    D=(D_{i,j})_{i\in[\ki],\,j\in[\ri]},
\]
where $Z_{i,j}\in\F_q^{\tau\times\sigma}$ and
$D_{i,j}\in\F_q^{\ell\times\sigma}$.
Since $\Aone$ and $\Mone$ are obtained by selecting,
respectively, the first $\sigma$ rows and the remaining
$\tau$ rows from each of the $\ki$ block rows of $F$,
each containing $\ell$ rows, we have
$D_{i,j}=
    \begin{bmatrix}
        E_{i,j}\\
        Z_{i,j}
    \end{bmatrix}.$
Choose a scalar Cauchy matrix $ Q\in\F_q^{\ki\tau\times\ri\tau},$
and partition it as $Q=(Q_{i,j})_{i\in[\ki],j\in[\ri]},$ where $ Q_{i,j}\in\F_q^{\tau\times\tau}.$
Define
\begin{equation}
\label{eq:I-Bij}
    B_{i,j}=
    \begin{bmatrix}
        E_{i,j} & 0\\
        Z_{i,j} & Q_{i,j}
    \end{bmatrix},
\end{equation}
and let $B=(B_{i,j})_{
        i\in[\ki],\,j\in[\ri]}.$
The initial code is generated by $ G^I=[I_{\ki\ell}\mid B].$ 
Figure~\ref{fig:caseI-schematic-compact} illustrates
the construction, and the same block-triangular structure is also used, up to coordinate permutations, in Cases~II and~III.

To verify the MDS property, let $I\subseteq[\ki]$ and
$J\subseteq[\ri]$ satisfy $|I|=|J|\ge1$.  
After permuting its rows and columns to group the first
$\sigma$ coordinates of each selected block before the
remaining $\tau$ coordinates, $B[I;J]$ becomes
    $\begin{bmatrix}
        E[I;J] & 0\\
        Z[I;J] & Q[I;J]
    \end{bmatrix}.$  
Both $E[I;J]$ and $Q[I;J]$ are square submatrices of scalar Cauchy matrices and are therefore nonsingular.  Thus, $B[I;J]$ is nonsingular, so $B$ is block superregular and the initial code is MDS.

\subsubsection{Conversion and field size}

Download the first $\sigma$ coordinates from each
retired parity node and read no systematic subsymbol.
The matrices in Lemma~\ref{lem:conversion-characterization} then satisfy $\widetilde C=F$ and $\widetilde B=D=FL.$
Since $L$ is invertible,
\[
    \col(\widetilde C)
    =
    \col(F)
    =
    \col(FL)
    =
    \col(\widetilde B).
\]
Hence,  Lemma~\ref{lem:conversion-characterization} ensures that these downloads permit a linear conversion
preserving all $\ki$ systematic symbols. 
Its read bandwidth is $R=\ri\sigma=\lambda\rf\ell,$ which attains the Case~I lower bound.

To construct the scalar Cauchy matrices $C$, $E$, and $Q$, it suffices to choose a prime power $q$ satisfying
\[ q>
    \max\bigl\{
        (\kf+\rf)\ell,\,
        (\ki+\ri)\sigma,\,
        (\ki+\ri)\tau
    \bigr\}.
\]
Hence one may take $ q=O\bigl((\ki+\ri)\ell\bigr).$

% =====================================================================画个图解释一下

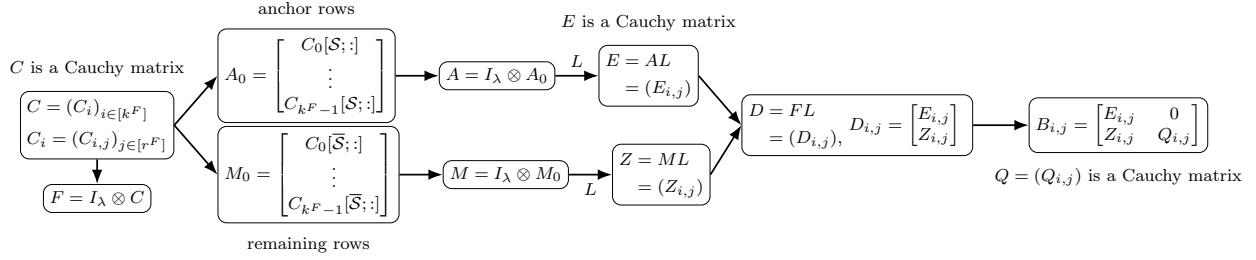
\begin{figure}[t]
\centering
\resizebox{\columnwidth}{!}{%
\begin{tikzpicture}[
    >=Latex,
    font=\scriptsize,
    node distance=6mm and 7mm,
    box/.style={
        draw,
        rounded corners,
        align=center,
        inner sep=3pt
    },
    arr/.style={->, thick}
]

% Left: C
\node[box] (C) {$\begin{aligned}
    &C=(C_i)_{i\in[\kf]}\\
    &C_i=(C_{i,j})_{j\in[\rf]}
\end{aligned}$};
\node[above=1mm of C] {\scriptsize $C$ is a  Cauchy matrix};
%%%%加上F的信息
\node[box, below=4mm of C] (F) {$
F=I_\lambda\otimes C
$};

\draw[arr] (C) -- (F);

% Top row
\node[box, right=of C, yshift=8mm] (A0) {$
A_0=
\begin{bmatrix}
C_0[\mathcal S;:]\\
\vdots\\
C_{\kf-1}[\mathcal S;:]
\end{bmatrix}
$};

\node[box, right=of A0] (A1) {$
A=I_\lambda\otimes A_0
$};

\node[box, right=of A1] (E) {$
\begin{aligned}
E&=AL\\
 &= (E_{i,j})
\end{aligned}
$};
\node[above=1mm of E] {\scriptsize $E$ is a  Cauchy matrix};

% Bottom row
\node[box, right=of C, yshift=-8mm] (M0) {$
M_0=
\begin{bmatrix}
C_0[\overline{\mathcal S};:]\\
\vdots\\
C_{\kf-1}[\overline{\mathcal S};:]
\end{bmatrix}
$};

\node[box, right=of M0] (M1) {$
M=I_\lambda\otimes M_0
$};

\node[box, right=of M1] (Z) {$
\begin{aligned}
Z&=ML\\
 &= (Z_{i,j})
\end{aligned}
$};

% Middle-right: D
\node[box, right=14mm of $(E)!0.5!(Z)$] (D) {$
\begin{aligned}
D&=FL\\
 &= (D_{i,j}),
\end{aligned}\ D_{i,j}=
\begin{bmatrix}
E_{i,j}\\
Z_{i,j}
\end{bmatrix}
$
};

% Rightmost: Bij
\node[box, right=9mm of D] (B) {$
B_{i,j}=
\begin{bmatrix}
E_{i,j} & 0\\
Z_{i,j} & Q_{i,j}
\end{bmatrix}
$};
\node[below=1mm of B] {\scriptsize $Q=(Q_{i,j})$ is a  Cauchy matrix};

% arrows
\draw[arr] (C.east)-- (A0.west);
\draw[arr] (C.east)--(M0.west);

\draw[arr] (A0) -- (A1);
\draw[arr] (A1) -- node[above] {$L$} (E);

\draw[arr] (M0) -- (M1);
\draw[arr] (M1) -- node[below] {$L$} (Z);

\draw[arr] (E.east) -- (D.west);
\draw[arr] (Z.east) -- (D.west);

\draw[arr] (D) -- (B);

% small labels
\node[above=1mm of A0] {\scriptsize anchor rows};
\node[below=1mm of M0] {\scriptsize remaining rows};

\end{tikzpicture}%
}
\caption{Construction of the initial parity blocks in Case~I.
The matrices $C$, $E$, and $Q$ are scalar Cauchy matrices.
An invertible matrix $L$ satisfying $\Aone L=E$
gives $Z=\Mone L$ and $D=FL$.
The blocks of $D$ are then completed with those of $Q$
to form $B_{i,j}$.}
\label{fig:caseI-schematic-compact}
\end{figure}
% =====================================================================

\subsection{Case II: $\kf<\ri<\ki$ and $\Gamma\ge1$}
\label{subsec:case-II}
Recall that $\Delta=\ri(\kf-\rf)+\ki\rf$.
In this case, the optimal read bandwidth is
$R=\frac{(\ki)^2\rf}{\Delta}\ell$.
Choose $\ell$ such that $\Delta\mid\ell$, and set
$\beta=\frac{\rf(\ki-\ri)}{\Delta}\ell$, $\sigma=\frac{\ki\rf}{\Delta}\ell.$
Each unchanged systematic node contributes $\beta$ subsymbols, while each retired parity node contributes $\sigma$ subsymbols.
Define
\[
   u=\ell-\beta=\frac{\ri\kf}{\Delta}\ell,\ 
    w=\sigma-\beta=\frac{\rf\ri}{\Delta}\ell,\ 
    \tau=\ell-\sigma=\frac{\ri(\kf-\rf)}{\Delta}\ell.
\]
Since $\rf<\kf$ and $\ri<\ki$, we have
$0<\beta<\sigma<\ell$ and $0<w<u$.
The parameters satisfy
\[
    u=w+\tau,
    \qquad
    \sigma=w+\beta,
    \qquad
    \kf w=\rf u,
    \qquad
    \ki w=\ri\sigma.
\]
In particular, $\lambda\rf u=\ri\sigma$.
\subsubsection{Construction of the final code $C^F$}
Choose a scalar Cauchy matrix $N =(n_{i,j})_{i\in[\kf],j\in[\rf]}\in\F_q^{\kf\times\rf},$ and set $\widehat C=N\otimes I_\ell.$
For any $I\subseteq[\kf]$ and $J\subseteq[\rf]$ with $|I|=|J|\ge 1$,  the matrix $\widehat C[I;J]=N[I;J]\otimes I_\ell$ is nonsingular.
Thus, $\widehat C$ is block superregular.
Before reordering the rows, designate the last $\beta$
coordinates in each systematic block as read coordinates.
Restricting each block row of $\widehat C$ to its first
$u$ rows and deleting the resulting last $\beta$ zero
columns from each block column gives
\[
    H=N\otimes I_u
      =
      \begin{bmatrix}
          H_0\\
          \vdots\\
          H_{\kf-1}
      \end{bmatrix},\ \text{where} \ 
    H_i=
    \begin{bmatrix}
        n_{i,0}I_u & \cdots & n_{i,\rf-1}I_u
    \end{bmatrix}\in\F_q^{u\times\rf u}.
\]
We select $w$ anchor rows from each $H_i$ by defining
$ S_i=\{iw,iw+1,\ldots,(i+1)w-1\}\pmod u,$ and $\overline S_i:=[u]\setminus S_i.$
Since $w<u$, the integers $iw,\ldots,(i+1)w-1$ have distinct remainders modulo $u$.
Hence $|S_i|=w$ and $|\overline S_i|=u-w=\tau$.
Set
$$\Atwo_0=\begin{bmatrix}
    H_0[S_0;:]\\ \vdots\\H_{\kf-1}[S_{\kf-1};:]
\end{bmatrix}\in\F_q^{\kf w\times\rf u},\ \Mtwo_0=\begin{bmatrix}
    H_0[\overline{S_0};:]\\ \vdots\\H_{\kf-1}[\overline{S_{\kf-1}};:]
\end{bmatrix}\in\F_q^{\kf\tau\times\rf u}. $$
To show that $\Atwo_0$ is nonsingular, for each $v\in[u]$ define $T_v:=\{i\in[\kf]:v\in S_i\}.$
The integers $iw+a$, with $i\in[\kf]$ and $a\in[w]$,
range over $[\kf w]=[\rf u]$.
Each residue modulo $u$ therefore occurs exactly $\rf$
times, in distinct block rows because $w<u$.
Hence $|T_v|=\rf$.
The $v$-th row of $H_i$ is supported on the $v$-th
coordinate of each block column, with entries
$n_{i,0},\ldots,n_{i,\rf-1}$.
Grouping the selected rows and columns by their coordinate index therefore transforms $\Atwo_0$ into $\operatorname{diag}\bigl(N[T_0;:],\ldots,N[T_{u-1};:]\bigr).$
Each diagonal block is a square submatrix of the Cauchy
matrix $N$ and is therefore nonsingular.
Thus, $\Atwo_0\in\mathrm{GL}_{\rf u}(\F_q)$.

For each $i\in[\kf]$, let $\Pi_i$ be the permutation matrix that orders the rows of the $i$-th block row of $\widehat C$ according to $S_i$, $[\ell]\setminus[u]$, and $\overline S_i$, with each set taken in increasing order.
Define $\Pi:=\operatorname{diag} (\Pi_0,\ldots,\Pi_{\kf-1}),$ and  $C=\Pi\widehat C.$
The rows in each systematic block are now ordered as
$w$ anchor coordinates, $\beta$ read coordinates, and
$\tau$ remaining unread coordinates.
In particular, the read coordinates occupy $[\sigma]\setminus[w]$.
Since $\Pi$ acts by an invertible row permutation within each block row, $C$ remains block superregular.
Thus, $G^F=[I_{\kf\ell}\mid C]$ generates an MDS array code.

Deleting the read rows from each block row of $C$
and the resulting zero columns gives
\[
    \widetilde H=
    \begin{bmatrix}
        H_0[S_0;:]\\
        H_0[\overline S_0;:]\\
        \vdots\\
        H_{\kf-1}[S_{\kf-1};:]\\
        H_{\kf-1}[\overline S_{\kf-1};:]
    \end{bmatrix}\in\F_q^{\kf u\times\rf u}.
\]
Collecting the first $w$ rows from each $u$-row block
of $\widetilde H$ gives $\Atwo_0$, while collecting
the remaining $\tau$ rows gives $\Mtwo_0$.
Define $F=I_\lambda\otimes\widetilde H\in\F_q^{\ki u\times\ri\sigma}$, $\Atwo=I_\lambda\otimes \Atwo_0\in \mathrm{GL}_{\ri\sigma}(\F_q),$ and  $\Mtwo:=I_\lambda\otimes\Mtwo_0
\in\F_q^{\ki\tau\times\ri\sigma}$.
These matrices will be used in the construction of the initial code.

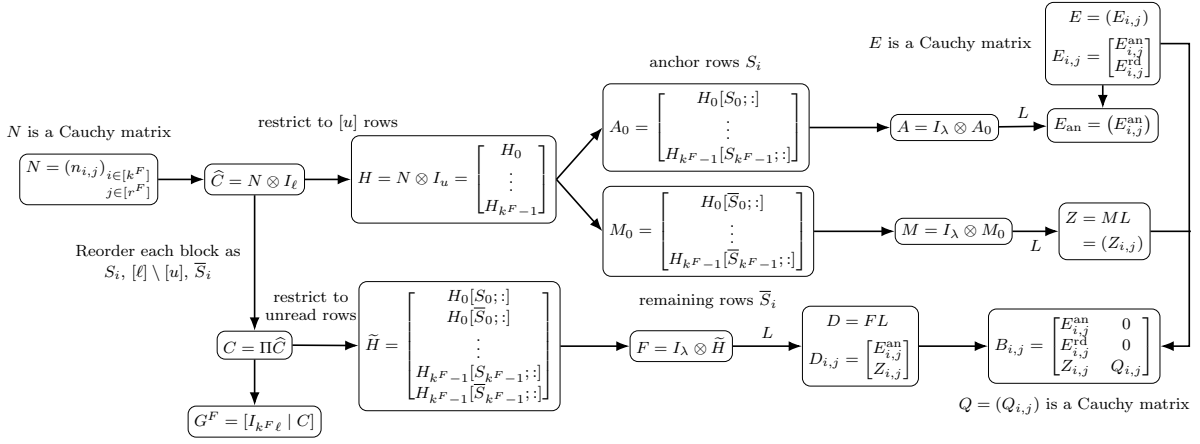
\begin{figure*}[t]
\centering
\resizebox{0.96\textwidth}{!}{%
\begin{tikzpicture}[
    >=Latex,
    font=\scriptsize,
    node distance=6mm and 8mm,
    box/.style={
        draw,
        rounded corners,
        align=center,
        inner sep=3pt
    },
    arr/.style={->, thick}
]

%====================================================
% Left part: construction of the final code
%====================================================

\node[box] (N) {$
N=(n_{i,j})_{\substack{i\in[\kf]\\ j\in[\rf]}}
$};
\node[above=1mm of N] {\scriptsize $N$ is a  Cauchy matrix};

\node[box, right=of N] (Chat) {$
\widehat C=N\otimes I_\ell
$};

\node[box, right=of Chat] (H) {$
H=N\otimes I_u=
\begin{bmatrix}
H_0\\
\vdots\\
H_{\kf-1}
\end{bmatrix}
$};

\draw[arr] (N) -- (Chat);
% \draw[arr] (Chat) -- (H);
\draw[arr] (Chat) -- node[above=7mm] {restrict to $[u]$ rows} (H);
% \node[above=1mm of Chat] {\scriptsize restrict to $[u]$ rows};

% A0 and M0
\node[box, right=of H, yshift=9mm] (A0) {$
\Atwo_0=
\begin{bmatrix}
H_0[S_0;:]\\
\vdots\\
H_{\kf-1}[S_{\kf-1};:]
\end{bmatrix}
$};

\node[box, right=of H, yshift=-9mm] (M0) {$
\Mtwo_0=
\begin{bmatrix}
H_0[\overline S_0;:]\\
\vdots\\
H_{\kf-1}[\overline S_{\kf-1};:]
\end{bmatrix}
$};

\draw[arr] (H.east) -- (A0.west);
\draw[arr] (H.east) -- (M0.west);

\node[above=1mm of A0] {\scriptsize anchor rows $S_i$};
\node[below=1mm of M0] {\scriptsize remaining rows $\overline S_i$};

% C and GF
\node[box, below=23.5mm of Chat] (C) {$
C=\Pi\widehat C
$};

\node[box, below=8mm of C] (GF) {$
G^F=[I_{\kf\ell}\mid C]
$};

%\draw[arr] (Chat) -- (C);
\draw[arr] (C) --(GF);
\draw[arr] (Chat) -- node[left=1mm] {\scriptsize \shortstack{Reorder each  block as\\$S_i$,  $[\ell]\setminus[u]$, $\overline S_i$}} (C);
%\node[below=1mm of C] {\scriptsize reorder as anchor | read | remaining};

% Htilde
\node[box, below=10mm of H, xshift=1mm] (Ht) {$
\widetilde H=
\begin{bmatrix}
H_0[S_0;:]\\
H_0[\overline S_0;:]\\
\vdots\\
H_{\kf-1}[S_{\kf-1};:]\\
H_{\kf-1}[\overline S_{\kf-1};:]
\end{bmatrix}
$};

% \draw[arr] (C.east) -- node[pos=0.8,above=3mm] {\scriptsize restrict to $[u]$ rows} (Ht);
\draw[arr] (C.east) -- node[pos=0.28, above=3mm] {\scriptsize \shortstack{restrict to\\ unread rows}} (Ht);

%====================================================
% Middle part: tensor replication and E,Z,D
%====================================================

\node[box, right=14mm of A0] (A1) {$
\Atwo=I_\lambda\otimes\Atwo_0
$};

\node[box, right=14mm of M0] (M1) {$
\Mtwo=I_\lambda\otimes\Mtwo_0
$};

\node[box, right=12mm of Ht] (F) {$
F=I_\lambda\otimes\widetilde H
$};

\draw[arr] (A0) -- (A1);
\draw[arr] (M0) -- (M1);
\draw[arr] (Ht) -- (F);

% E and E_an
\node[box,above right=5mm and 8mm of A1] (E) {$
\begin{aligned}
E&=(E_{i,j})\\
E_{i,j}&=
\begin{bmatrix}
E^{\rm an}_{i,j}\\
E^{\rm rd}_{i,j}
\end{bmatrix}
\end{aligned}
$};
\node[left=1mm of E] {\scriptsize $E$ is a  Cauchy matrix};
\node[box, below=4mm of E] (Ean) {$
E_{\rm an}=
\bigl(E^{\rm an}_{i,j}\bigr)
$};

\draw[arr] (E.south) -- (Ean.north);

% Z
\node[box, right=of M1] (Z) {$
\begin{aligned}
Z&=\Mtwo L\\
 &=(Z_{i,j})
\end{aligned}
$};

% D
\node[box, right=12mm of F] (D) {$
\begin{aligned}
D&=FL\\
D_{i,j}&=
\begin{bmatrix}
E^{\rm an}_{i,j}\\
Z_{i,j}
\end{bmatrix}
\end{aligned}
$};

% Arrows with L
\draw[arr] (A1) -- node[above] {$L$} (Ean);
\draw[arr] (M1) -- node[below] {$L$} (Z);
\draw[arr] (F) -- node[above] {$L$} (D);

%====================================================
% Right part: B_ij
%====================================================

\node[box, right=12mm of D] (B) {$
B_{i,j}=
\begin{bmatrix}
E^{\rm an}_{i,j} & 0\\
E^{\rm rd}_{i,j} & 0\\
Z_{i,j} & Q_{i,j}
\end{bmatrix}
$};

\draw[arr] (E.east) -- ++(5mm,0) |- (B.east);
% \draw[arr] (Z.east) -- ++(5mm,0) ;
\draw[arr, -]  (Z.east) -- ++(7mm,0);
\draw[arr] (D.east) -- (B.west);

% Optional label for Q
\node[below=1mm of B] {\scriptsize $Q=(Q_{i,j})$ is a Cauchy  matrix};

\end{tikzpicture}%
}
\caption{
Construction of the initial and final codes in Case~II.
The matrices $N$, $E$, and $Q$ are scalar Cauchy matrices.
The row permutation $\Pi$ orders each systematic block
as anchor, read, and remaining unread coordinates.
Setting $L=\Atwo^{-1}E_{\rm an}$ gives $Z=\Mtwo L$
and $D=FL$, from which the initial parity blocks
are formed as in \eqref{eq:II-Bij}.
}
\label{fig:caseII-complete}
\end{figure*}
\subsubsection{Construction of the initial code $C^I$}
Choose a scalar Cauchy matrix $E\in\F_q^{\ki\sigma\times\ri\sigma},$ and partition it into $\ki\times\ri$ blocks $E=(E_{i,j})_{i\in[\ki],j\in[\ri]}$, where $E_{i,j}\in\F_q^{\sigma\times\sigma}.$
Since $\sigma=w+\beta$, write
\[
    E_{i,j}=
    \begin{bmatrix}
        E^{\rm an}_{i,j}\\
        E^{\rm rd}_{i,j}
    \end{bmatrix},
\]
where $E^{\rm an}_{i,j}\in\F_q^{w\times\sigma}$,  $E^{\rm rd}_{i,j}\in\F_q^{\beta\times\sigma}.$
Let $E_{\rm an}$ be the matrix obtained by retaining
the first $w$ rows of each block row of $E$.
Since $\ki w=\ri\sigma,$ it is a square submatrix
of the scalar Cauchy matrix $E$ and is therefore
nonsingular.
Set $L=\Atwo^{-1}E_{\rm an}$, $Z=\Mtwo L$ and $D=FL.$
Then $L$ is invertible and $\Atwo L=E_{\rm an}$.
Partition $Z$ and $D$ as
 $Z=(Z_{i,j})_{i\in[\ki],j\in[\ri]}$, $D=(D_{i,j})_{i\in[\ki],j\in[\ri]}$, where $Z_{i,j}\in\F_q^{\tau\times\sigma}$ and
$D_{i,j}\in\F_q^{u\times\sigma}$.
Since $\Atwo$ and $\Mtwo$ collect the first $w$
and the remaining $\tau$ rows from each of the $\ki$
block rows of $F$, respectively, we have
\[
    D_{i,j}
    =
    \begin{bmatrix}
        E^{\rm an}_{i,j}\\
        Z_{i,j}
    \end{bmatrix}.
\]

Choose a scalar Cauchy matrix $Q\in\F_q^{\ki\tau\times\ri\tau}$
and partition it as $ Q=(Q_{i,j})_{i\in[\ki],j\in[\ri]},$ where $Q_{i,j}\in\F_q^{\tau\times\tau}.$
For each $i\in[\ki]$ and $j\in[\ri]$, define
\begin{equation}
\label{eq:II-Bij}
    B_{i,j}
    :=
    \begin{bmatrix}
        E^{\rm an}_{i,j} & 0\\
        E^{\rm rd}_{i,j} & 0\\
        Z_{i,j}          & Q_{i,j}
    \end{bmatrix}=
    \begin{bmatrix}
        E_{i,j} & 0\\
        Z_{i,j} & Q_{i,j}
    \end{bmatrix},
\end{equation}
and let $B=(B_{i,j})_{i\in[\ki],\,j\in[\ri]}.$
The initial code $C^I$ is generated by $G^I=[I_{\ki\ell}\mid B].$
Since $E$ and $Q$ are scalar Cauchy matrices,
the same block-triangular argument as in Case~I
shows that $B$ is block superregular and the initial
code is MDS.
Figure~\ref{fig:caseII-complete} illustrates the
constructions of the initial and final codes.

\subsubsection{Conversion and field size}
Read the $\beta$ coordinates indexed by
$[\sigma]\setminus[w]$ from every unchanged systematic
node, and download the first $\sigma$ coordinates
from every retired parity node.
Deleting the read rows and retaining the downloaded
parity columns in $B$ gives $\widetilde B=D=FL.$
On the final-code side, deleting the zero columns
of $\widetilde C$ gives $F$, so $\col(\widetilde C)=\col(F)$.
Since $L$ is invertible,
\[
    \col(\widetilde C)
    =\col(F)
    =\col(FL)
    =\col(\widetilde B).
\]
Thus, Lemma~\ref{lem:conversion-characterization}
ensures that these downloads permit a linear conversion
preserving all $\ki$ systematic symbols.
Its read bandwidth is
$R=\ki\beta+\ri\sigma=\frac{(\ki)^2\rf}{\Delta}\ell,$
which attains the Case~II lower bound.

To construct the scalar Cauchy matrices $N$, $E$, and $Q$, it suffices to choose a prime power $q$ satisfying $q>
    \max\bigl\{
        \kf+\rf,\,
        (\ki+\ri)\sigma,\,
        (\ki+\ri)\tau
    \bigr\}$.  
Hence one may take $ q=O\bigl((\ki+\ri)\ell\bigr).$

% =====================================================================

\subsection{Case III: $\kf<\ri<\ki$ and $\Gamma<1$}
\label{subsec:case-III}

Recall that $\mthr=\left\lfloor\frac{\ri-1}{\kf}\right\rfloor,$ $\delthr=\ri-\mthr\kf,$ $ \hthr=\lambda-\mthr, $
and $\Tthr=\ri+(\hthr-1)\rf.$
Choose $\ell$ such that $\delta T\mid\ell$, and set
$\betthr=\frac{(\hthr-1)\rf}{\Tthr}\ell$, $\sigthr=\frac{\rf(\mthr\kf+\hthr\delthr)}{\delthr \Tthr}\ell$. 
Each unchanged systematic node contributes $\beta$
subsymbols, while each retired parity node contributes
$\sigma$ subsymbols.
These read amounts attain the Case~III lower bound
once the conversion is established.
Define $\uthr=\ell-\beta=\frac{\ri}{T}\ell,$ $ \wthr=\sigma-\beta=\frac{\rf}{\delta}u,$ and $\tauthr=\ell-\sigthr=u-w.$
The condition $\Gamma<1$ implies $h\ge2$ and
$\rf<\delta\le\kf$, so $0<\beta<\sigma<\ell$.
The following identities will be used below:
\[
    \delthr \wthr=\rf \uthr,
    \qquad
    \ri \betthr=(\hthr-1)\delthr\wthr,
    \qquad
    \ri\sigthr=(\mthr\kf+\hthr\delthr)\wthr.
\]

\subsubsection{Construction of the final code $C^F$}

Set $ \nzerthr=\kf\wthr,$ $\dzerthr=\delthr\wthr=\rf\uthr.$ 
Then $d_0\le n_0$.
Choose  a multiplicative subgroup $\mathHthr$ of $\F_q$ of size $\nzerthr$, and partition it as 
\[
\mathHthr=\mathHthr_0\mathbin{\dot\cup}\cdots\mathbin{\dot\cup}\mathHthr_{\kf-1},
    \ \text{where }  |\mathHthr_i|=w\  \text{for} \ i\in[\kf].
\]
Choose two disjoint subsets $\mathUthr,\mathVthr\subset \F_q\setminus\mathHthr$ with 
$|\mathUthr|=\kf\tau$ and $|\mathVthr|=d_0$, and partition them as
\begin{align*}
\mathUthr=\mathUthr_0\mathbin{\dot\cup}\cdots\mathbin{\dot\cup}\mathUthr_{\kf-1},
    \ \text{where }  |\mathUthr_i|=\tauthr\ \text{for} \ i\in[\kf],\\
\mathVthr=\mathVthr_0\mathbin{\dot\cup}\cdots\mathbin{\dot\cup}\mathVthr_{\rf-1},
    \ \text{where }  |\mathVthr_j|=\uthr\ \text{for} \ j\in[\rf].
\end{align*}
For $i\in[\kf],j\in[\rf]$, define
$\Wthr_{i,j}=(\frac{1}{a-b})_{
        a\in\mathHthr_i\cup\mathUthr_i,\,
        b\in\mathVthr_j}\in\F_q^{\uthr\times\uthr}$,
using fixed orderings within each set and placing the rows indexed by $\mathcal H_i$ before those indexed by $\mathcal U_i$.
 Then $\Wthr=(\Wthr_{i,j})_{i\in[\kf],j\in[\rf]}\in\F_q^{\kf u\times\rf u}$ is a scalar Cauchy matrix.
Choose another scalar Cauchy matrix $ \Qthr\in\F_q^{\kf\betthr\times\rf\betthr},$ and partition it into $\kf\times\rf$ blocks $\Qthr=(\Qthr_{i,j})_{i\in[\kf],j\in[\rf]}$.
Define $C_{i,j}=\begin{bmatrix}
    \Wthr_{i,j}&0\\
    0&\Qthr_{i,j}
\end{bmatrix}$, and let $C=(C_{i,j})_{i\in[\kf],j\in[\rf]}$.
After row and column permutations, every square block
submatrix of $C$ becomes block diagonal, with diagonal
blocks that are square submatrices of $\Wthr$ and
$\Qthr$.
Since both diagonal blocks are nonsingular, $C$ is block superregular, and $G^F=[I_{\kf\ell}\mid C]$ generates an MDS array code.
% The last $\beta$ rows in $\begin{bmatrix} 0&\Qthr_{i,j} \end{bmatrix}$ is designed to be the removed $\beta$ rows. 

For each $i\in[\kf]$, let the rows of $\Wthr_{i,j}$ indexed by $\mathHthr_i$ be its $\wthr$ anchor rows.
Let $\Wthr_{\rm an}=(\Wthr_{i,j}[\mathHthr_i;:])_{i\in[\kf],j\in[\rf]}\in
    \F_q^{\kf\wthr\times\dzerthr}$ be  the matrix formed by  all anchor rows, and 
let $\Wthr_{\rm lo}=(\Wthr_{i,j}[\mathUthr_i;:])_{i\in[\kf],j\in[\rf]}\in
    \F_q^{\kf\tauthr\times\dzerthr}$ be the matrix
    formed by  the remaining rows. 
% After a fixed row permutation, $\Wthr$ is equivalent to
%     $\begin{bmatrix}
%         \Wthr_{\rm an}\\
%         \Wthr_{\rm lo}
%     \end{bmatrix}.$
We next express $\col(\Wthr_{\rm an})$ in terms of
polynomial evaluations.
Define $\Qpolthr=\prod_{b\in \mathVthr}(X-b).$ 
For each $b\in\mathVthr$, the column of $\Wthr_{\rm an}$ corresponding to $b$ can be written as  $\left( Q_{\mathcal V}(x)^{-1}\left(\frac{Q_\mathcal{V}(x)}{x-b}\right)\right)_{x\in\mathHthr}$. 
The polynomials $\left\{
        \frac{\Qpolthr}{X-b}
        :
        b\in\mathVthr
    \right\}$ are linearly independent.
Otherwise, if
    $\sum_{b\in\mathVthr}
    c_b
    \frac{\Qpolthr}{X-b}
    =0,$ then evaluating at $X=b_0\in\mathVthr$ gives $ c_{b_0}\prod_{\substack{b\in\mathVthr\\b\ne b_0}}(b_0-b)=0,$ and hence $c_{b_0}=0$.
Since these $\dzerthr$ polynomials all have degree less than $\dzerthr$, they form a basis of the polynomial space of degree less than $\dzerthr$.
Therefore, we have 
\begin{equation}
\label{eq:III-Cauchy-GRS}
    \col(\Wthr_{\rm an})
    =
    \left\{
        \left(Q_\mathcal{V}(x)^{-1}p(x)\right)_{x\in H}:
        \deg p< \dzerthr
    \right\}.
\end{equation}
After the last $\beta$ read coordinates are removed, the residual final parity transform for one final group is ${\Wthr}$. 
Define $\Fthr=I_\lambda\otimes\Wthr,$  $\Fthr_{\rm an}=I_\lambda\otimes\Wthr_{\rm an},$ and  $\Fthr_{\rm lo}=I_\lambda\otimes\Wthr_{\rm lo}.$
Then $\Fthr_{\rm an}$ and $\Fthr_{\rm lo}$ are respectively the anchor-row and remaining-row restrictions of $\Fthr$.
% Define $\Fthr=I_\lambda\otimes \widetilde{  \Wthr},$
% and $\Fthr_{\rm an}:=I_\lambda\otimes   \Wthr_{\rm an},$    $\Fthr_{\rm lo}:=I_\lambda\otimes   \Wthr_{\rm lo}.$
% Then we have, $\Fthr=\begin{bmatrix}\Fthr_{\rm an}\\ \Fthr_{\rm lo} \end{bmatrix}$.
    
\subsubsection{Construction of the initial code $C^I$}
We first construct an auxiliary
$[\ki+\ri,\ri,\sigma]$ MDS array code $\mathcal E$.
Its systematic representation will provide a block
superregular matrix $E$ whose first $w$ rows in each
block row form a matrix $E_{\rm an}$ satisfying
$\col(\Fthr_{\rm an})\subseteq\col(E_{\rm an}).$
Choose $\lambda$ distinct cosets $\zeta_0\mathHthr,\ldots,\zeta_{\lambda-1}\mathHthr$ in $\F_q^\times$, and define $g(X):=X^{n_0},$ and  $\theta_t=\zeta_t^{n_0}$ for $t\in[\lambda].$
For every $x\in\mathcal H$, we have $$g(\zeta_t x)=(\zeta_t x)^{n_0}=\zeta_t^{n_0}=\theta_t.$$
Thus $g$ is constant on each coset $\zeta_t\mathcal H$.
Moreover, the values $\theta_0,\ldots,\theta_{\lambda-1}$ are pairwise distinct. 
This is because $\mathcal H=\{x\in\F_q^\times:x^{n_0}=1\}$ for $|\mathcal H|=n_0$.
Thus $\theta_s=\theta_t$ would imply $\zeta_s/\zeta_t\in\mathcal H$,  contradicting the choice of
distinct cosets.

Define the polynomial space
\begin{equation}
\label{eq:III-poly-space-new_1}
\mathcal P:=\Bigl\{\sum_{j=0}^{m-1}g(X)^jp_j(X)+\sum_{j=m}^{m+h-1}g(X)^jq_j(X):\deg p_j<n_0,\ \deg q_j<d_0
\Bigr\}.
\end{equation}
Since $g(X)=X^{n_0}$ and $d_0\le n_0$, the monomials
in different summands of \eqref{eq:III-poly-space-new_1} have distinct degrees.
Therefore,
\[
    \dim\mathcal P
    =
    mn_0+hd_0
    =
    (m\kf+h\delta)w
    =
    \ri\sigma.
\] 
Moreover, every polynomial in $\mathcal P$ has degree at most
\[
    (\lambda-1)\nzerthr+\dzerthr-1
    <
    \lambda\nzerthr
    =
    \ki\wthr.
\]
We now use $\mathcal P$ to  construct the auxiliary code $\mathcal E$. 
For a message $f\in\mathcal P$,  each auxiliary node stores $ \sigthr= \wthr+\betthr$ subsymbols in two layers.  
In the first layer, node $s\kf+t$, for
$s\in[\lambda]$ and $t\in[\kf]$, stores $\left(Q_{\mathcal V}(x)^{-1}f(\zeta_sx)\right)_{x\in\mathcal H_t}.$ 
The coordinates follow the ordering of $\mathHthr_t$
used in $\Wthr$.
The first $\ki$ auxiliary nodes therefore use $\ki w$ distinct evaluation points.
All multipliers are nonzero because $\mathHthr\cap\mathVthr=\varnothing$.
For the last $\ri$  auxiliary nodes, choose $\ri w$ distinct elements $$\mathFthr =\mathFthr_0\mathbin{\dot\cup}\cdots\mathbin{\dot\cup}\mathFthr_{\ri-1}\subseteq\F_q\setminus \bigcup_{s\in[\lambda]}\zeta_s\mathHthr,$$ where $|\mathFthr_i|=\wthr$ for $i\in[\ri].$  
Node $\ki+i$ stores $\left(f(x)\right)_{x\in\mathFthr_i}$
in its first layer, using a fixed ordering of
$\mathFthr_i$.
Thus, all $(\ki+\ri)w$ first-layer evaluation points
are distinct.

For the second layer, collect the coefficients of
$X^0,\ldots,X^{d_0-1}$ in each of $q_{m+1},\ldots,q_{m+h-1}$.
There are
    $(h-1)d_0=(h-1)\delta w=\ri\beta$
such coefficients.
Encode them using a scalar $[(\ki+\ri)\beta,\ri\beta]$ Reed--Solomon code and place $\beta$ encoded symbols in each auxiliary node. 
Together, the two layers give $\sigma$ subsymbols
per node, as illustrated in Figure~\ref{fig:caseIII-E-encoding}.
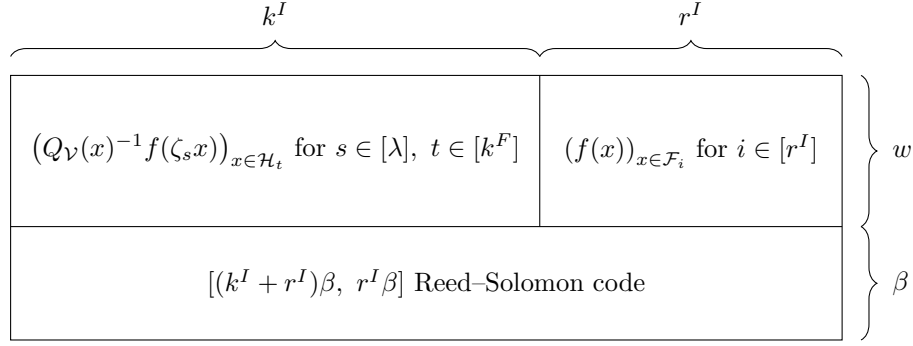
\begin{figure}[t]
\centering
\begin{tikzpicture}[scale=1.0, every node/.style={font=\small}]

\draw (0,0) rectangle (11,3.5);
\draw (0,1.5) -- (11,1.5);
\draw (7,1.5) -- (7,3.5);

\node[align=center] at (3.5,2.5)
{$\left(Q_{\mathcal V}(x)^{-1}f(\zeta_s x)\right)_{x\in\mathcal H_t}$ for
$s\in[\lambda],\ t\in[\kf]$};
% {$\left(Q_{\mathcal V}(x)^{-1}f(\zeta_s x)\right)_{x\in\mathcal H_t}$\\[0.4em]
% $s\in[\lambda],\ t\in[\kf]$};

\node[align=center] at (9,2.5)
{$\left(f(x)\right)_{x\in\mathcal F_i}$ for 
$i\in[\ri]$};

\node[align=center] at (5.5,0.75)
{$[(\ki+\ri)\beta,\ \ri\beta]$ Reed--Solomon  code};

\draw[decorate,decoration={brace,amplitude=6pt}]
(0,3.75) -- (7,3.75)
node[midway,above=8pt] {$\ki$};

\draw[decorate,decoration={brace,amplitude=6pt}]
(7,3.75) -- (11,3.75)
node[midway,above=8pt] {$\ri$};

\draw[decorate,decoration={brace,amplitude=6pt,mirror}]
(11.25,1.5) -- (11.25,3.5)
node[midway,right=8pt] {$w$};

\draw[decorate,decoration={brace,amplitude=6pt,mirror}]
(11.25,0) -- (11.25,1.5)
node[midway,right=8pt] {$\beta$};

\end{tikzpicture}
\caption{Encoding scheme of the $[\ki+\ri,\ri,\sigma]$ MDS array code $\mathcal E$ in Case~III, where $f\in\mathcal P$. 
The first layer consists of weighted evaluations on the first $\ki$  auxiliary nodes and ordinary evaluations on the last $\ri$ auxiliary nodes, while the second layer is a Reed--Solomon encoding of dimension $\ri\betthr$.}
\label{fig:caseIII-E-encoding}
\end{figure}
To verify the MDS property of $\mathcal E$, consider
any $\ri$ auxiliary nodes.
Their $\ri\beta$ second-layer symbols recover
$q_{m+1},\ldots,q_{m+h-1}$.
Subtracting the corresponding contributions from the
first-layer symbols and dividing by the nonzero
evaluation multipliers gives evaluations of
\[
    f_0(X)
    =
    \sum_{j=0}^{m-1}X^{jn_0}p_j(X)
    +
    X^{mn_0}q_m(X).
\]
Since
\[
    \deg f_0
    <
    mn_0+d_0
    =
    (m\kf+\delta)w
    =
    \ri w,
\]
the $\ri w$ distinct  evaluation points supplied by these $\ri$ nodes uniquely determine $f_0$.  
Hence any $\ri$ nodes recover $f$, proving that
$\mathcal E$ is an $[\ki+\ri,\ri,\sigma]$
MDS array code.
Use the last $\ri$ auxiliary nodes as an information
set and write the systematic encoding matrix in
column form as
\[
    M_{\mathcal E}
    =
    \begin{bmatrix}
        E\\
        I_{\ri\sigma}
    \end{bmatrix},
    \qquad
    E\in\F_q^{\ki\sigma\times\ri\sigma}.
\]
For $\mathbf m\in\F_q^{\ri\sigma}$, the corresponding
auxiliary codeword, written as a column vector,
is $M_{\mathcal E}\mathbf m$.
The MDS property implies that $E$ is block
superregular with $\sigma\times\sigma$ blocks.

Let $E_{\rm an}\in\F_q^{\ki w\times\ri\sigma}$
be obtained by retaining the first $w$ rows of
each block row of $E$.
By construction, the column space $\col(E_{\rm an})$ consists of the weighted evaluation vectors of polynomials in $\mathcal P$ on the first $\ki$ auxiliary nodes.
Consequently, $\col(E_{\rm an})$ is the space of
first-layer evaluation vectors on the first
$\ki$ auxiliary nodes.
Since every polynomial in $\mathcal P$ has degree less than $\ki\wthr$, evaluation on the $\ki\wthr$ distinct points of the first $\ki$ nodes is injective.
Therefore $\rank E_{\rm an}=\ri\sigma.$

We next show that $\col(\Fthr_{\rm an})\subseteq\col(E_{\rm an}).$
For $t\in[\lambda]$, define the Lagrange polynomial
\[
    L_t(Y):=
    \prod_{\substack{s\in[\lambda]\\s\ne t}}
    \frac{Y-\theta_s}{\theta_t-\theta_s}.
\]
For every polynomial $p$ with $\deg p<d_0$, define $f_{t,p}(X)=L_t(g(X))p(X/\zeta_t).$
Expanding $L_t$ shows that $f_{t,p}$ is a linear
combination of $g(X)^jp(X/\zeta_t)$, $j\in[\lambda]$.
Since $d_0\le n_0$, each term satisfies the degree
conditions in \eqref{eq:III-poly-space-new_1};
hence $f_{t,p}\in\mathcal P$.
For $x\in\mathHthr$, we have
\[
    f_{t,p}(\zeta_sx)
    =
    \begin{cases}
        p(x), & s=t,\\
        0, & s\ne t.
    \end{cases}
\]
By \eqref{eq:III-Cauchy-GRS}, the corresponding
weighted evaluation vectors range over $\col(\Wthr_{\rm an})$ on the $t$-th coset
and vanish on all other cosets.
Varying $t$ and taking linear combinations gives $\col(F_{\rm an})\subseteq\col(E_{\rm an}).$
This inclusion and the full column rank of
$E_{\rm an}$ imply that there is a unique matrix
$J\in\F_q^{\ri\sigma\times\lambda d_0}$ satisfying
    $E_{\rm an}J=\Fthr_{\rm an}.$
Since $\Fthr_{\rm an}$ has full column rank,
so does $J$.
Choose a left inverse
$K\in\F_q^{\lambda d_0\times\ri\sigma}$
satisfying $KJ=I_{\lambda d_0}$, and define
    $Z=\Fthr_{\rm lo}K
      \in\F_q^{\ki\tau\times\ri\sigma}.$
Then we have  $ZJ=\Fthr_{\rm lo}.$
The matrices $J$ and $K$ can be computed by
Gaussian elimination.
Partition $E$ and $Z$ into $\ki\times\ri$ blocks and write $E=(E_{i,j})_{i\in[\ki],j\in[\ri]},$ where $ E_{i,j}\in\F_q^{\sigma\times\sigma},$ and  $Z=(Z_{i,j})_{i\in[\ki],j\in[\ri]},$  where  $Z_{i,j}\in\F_q^{\tau\times\sigma}.$
Write $E_{i,j}=\begin{bmatrix}E^{\rm an}_{i,j}\\E^{\rm rd}_{i,j}\end{bmatrix},$
where $E^{\rm an}_{i,j}\in\F_q^{\wthr\times\sigthr},$ and $E^{\rm rd}_{i,j}\in\F_q^{\betthr\times\sigthr}.$
Choose a scalar Cauchy matrix $\Qthrini\in\F_q^{\ki\tau\times\ri\tau},$ and partition it into $\ki\times\ri$ blocks as $\Qthrini=(\Qthrini_{i,j})_{i\in[\ki],j\in[\ri]},$  where $\Qthrini_{i,j}\in\F_q^{\tau\times\tau}.$
For $i\in[\ki]$ and $j\in[\ri]$, define
\begin{equation}
\label{eq:III-Bij-new}
    B_{i,j}=
    \begin{bmatrix}
         E^{\rm an}_{i,j} & 0\\
        Z_{i,j} & \Qthrini_{i,j}\\
        E^{\rm rd}_{i,j}& 0
    \end{bmatrix},
\end{equation}
and let $ B=(B_{i,j})_{i\in[\ki],\,j\in[\ri]}.$
The initial code $C^I$ is generated by $ G^I=[I_{\ki\ell}\mid B].$

Moving the last $\beta$ rows of each block row
immediately after its first $w$ rows gives blocks
of the form
\[
    \begin{bmatrix}
        E_{i,j} & 0\\
        Z_{i,j} & \Qthrini_{i,j}
    \end{bmatrix}.
\]
Since $E$ is block superregular and $\Qthrini$
is a scalar Cauchy matrix, the block-triangular
argument from Case~I proves that $B$ is block
superregular.
Thus, the initial code is MDS.

\subsubsection{Conversion and field size}
\label{caseIII_field}

Read the last $\betthr$ coordinates from every unchanged systematic node and download the first $\sigthr$ coordinates from every retired parity node.
For these read sets, the matrix $\widetilde C$ in
Lemma~\ref{lem:conversion-characterization} reduces to
$\Fthr$ after its zero columns are deleted.
Hence $\col(\widetilde C)=\col(\Fthr)$.
Let $P$ be the row permutation matrix that places all anchor rows before the remaining unread rows, preserving their relative order within each part.
Then
\[
    P\widetilde B
    =
    \begin{bmatrix}
        E_{\rm an}\\
        Z
    \end{bmatrix},
    \qquad
    P\Fthr
    =
    \begin{bmatrix}
        \Fthr_{\rm an}\\
        \Fthr_{\rm lo}
    \end{bmatrix}.
\]
Since $E_{\rm an}J=\Fthr_{\rm an}$ and
$ZJ=\Fthr_{\rm lo}$, we have
$P\widetilde B J=P\Fthr$, and hence
$\widetilde B J=\Fthr$.
Therefore, $\col(\widetilde C)=\col(\Fthr)\subseteq\col(\widetilde B).$
By Lemma~\ref{lem:conversion-characterization},
these downloads permit a linear conversion preserving all $\ki$ systematic symbols.
Its read bandwidth is
\[
\begin{aligned}
    R
    &=\ki\beta+\ri\sigma\\
    &=
    \left(
        \frac{\ki(h-1)\rf}{T}
        +\frac{\ri\rf(m\kf+h\delta)}{\delta T}
    \right)\ell,
\end{aligned}
\]
which attains the Case~III lower bound.

For the field size, it suffices to choose a prime power
$q$ satisfying $\kf w\mid(q-1)$ and $q>M_{\rm III}$,
where
$    M_{\rm III}
    =
    \max\bigl\{
        (\ki+\ri)w,\,
        (\kf+\rf)u,\,
        (\ki+\ri)\beta,\,
        (\ki+\ri)\tau
    \bigr\}.$
These conditions ensure the existence of $\mathHthr$
and enough distinct cosets, evaluation points, and
Cauchy parameters for the construction.

To obtain an unconditional polynomial field-size bound,
set $N_0=\kf w\left\lceil\frac{M_{\rm III}+1}{\kf w}\right\rceil.$
By Lemma~\ref{lem:linnik},
there exists a prime $q\equiv1\pmod{N_0}$ with
$q=O(N_0^{L_0})$ for an absolute constant $L_0$.
Since $\kf w\mid N_0$ and $q>N_0>M_{\rm III}$,
this prime satisfies the required field conditions.
Moreover, $N_0=O((\ki+\ri)\ell)$, so one may take
$    q=O\!\left(((\ki+\ri)\ell)^{L_0}\right).$

\section{Conclusion}
\label{sec:conclusion}

In this paper, we studied the  bandwidth cost of stable MDS convertible codes with linear conversion procedures in the split regime for the remaining parameter range $ \rf<\kf<\ri.$

By considering coordinate projections of the associated column spaces of the initial and final codes onto different groups of coordinates, we derive a family of rank inequalities involving the amounts of data
read from the unchanged and retired nodes.
These inequalities recover the previous lower bounds for the regime $\ri\le\kf$ with $\rf<\kf$.
When $\rf<\kf<\ri$, the inequalities impose further constraints that strengthen the resulting lower bound.
By solving the corresponding linear program, we obtain an improved  lower bound for this remaining regime, consisting of three cases.

To establish the tightness of the lower bound, we presented three explicit constructions.  
The constructions for Cases~I and~II are based on selecting suitable \emph{anchor rows} of the final parity matrix, and choosing an invertible transformation that maps the corresponding anchor-row matrix to a Cauchy submatrix of the initial parity matrix.
This yields the required column-space relation for conversion, while the resulting Cauchy structure guarantees the MDS property.
For Case~III, to realize the required column-space inclusion for conversion while preserving the MDS property, we construct an auxiliary MDS array code with two layers: the first layer is based on structured polynomial evaluations, and the second layer is based on a Reed--Solomon code.
The three constructions attain the corresponding lower bounds.
The  required integrality conditions can always be satisfied by choosing the subpacketization $\ell$ appropriately, with $\ell$ polynomial in the the initial code length $n^I$.
The constructions for Cases~I and~II can be realized over fields of size $ q=O\bigl((\ki+\ri)\ell\bigr),$
whereas for Case~III it suffices to take $q=O\!\left(((\ki+\ri)\ell)^{L_0}\right)$ for some absolute constant $L_0$.
Consequently, our lower bound is tight throughout the regime $\rf<\kf<\ri$ within the stable linear conversion model.
Since the write bandwidth is fixed for stable conversion, our results determine the minimum total conversion bandwidth in this regime.

Several questions remain for future work. 
It would be desirable to reduce the field size required by all three constructions to be linear in the code length.
Another direction is to determine the exact optimum for a prescribed subpacketization when the LP-optimal read amounts are nonintegral, without enlarging $\ell$.
\bibliographystyle{IEEEtran}
\bibliography{ref}
\end{document}